\documentclass[11pt,a4paper]{article}
\usepackage{amsmath,amssymb,amsthm,newlfont,graphicx,amscd,bbm,enumerate,galois,mathrsfs,xypic,geometry}
\usepackage{simplewick}
\usepackage{tcolorbox}
\usepackage{indentfirst,hyperref}
\usepackage{booktabs,caption}   
\usepackage{setspace}
\usepackage{color}

\newcommand{\lmd}{\lambda}

\newcommand{\p}{\partial}
      
\newcommand{\afa}{\alpha}

\newcommand{\veps}{\varepsilon}

\newcommand{\rme}{{\mathrm e}}

\newcommand{\mcalA}{\mathcal{A}}

\newcommand{\mcalP}{\mathcal{P}}

\newcommand{\bsP}{\boldsymbol{P}}

\DeclareMathOperator{\Res}{Res}
\newcommand*{\pp}[1]
{\frac{\partial   }
	{\partial #1}
}

\newcommand*{\pfrac}[2]
{\frac{\partial #1}
	{\partial #2}
}

\newcommand*{\pair}[2]                   
{
	\left\langle
	#1,#2
	\right\rangle
}

\newcommand*{\Bigset}[2]
{
	\left\{ #1 \,\middle|\, #2 \right\}             
}

\newcommand{\beq}{\begin{equation}}
	\newcommand{\eeq}{\end{equation}}

\DeclareMathOperator{\res}{Res}

\DeclareMathOperator{\grad}{grad}

\DeclareMathOperator{\td}{d\!}

\newtheorem{thm}{Theorem}[section]

\newtheorem{rmk}[thm]{Remark}

\newtheorem{lem}[thm]{Lemma}

\newtheorem{prop}[thm]{Proposition}
\newtheorem{defn}{Definition}[section]
\newtheorem{ex}{Example}[section]

\numberwithin{equation}{section}

\allowdisplaybreaks[4]

\begin{document}
	\title{{Tri-Hamiltonian structure of an asymmetric generalized Ablowitz–Ladik hierarchy and a Frobenius manifold}
				\footnotetext{{\footnotesize$^*$Correspondence should be addressed to Qiulan Zhao; } \\
			{\footnotesize Electronic mail:} {\footnotesize
				 qlzhao@sdust.edu.cn}}}
				\author{{
			Hanxue Zhang,~ Qiulan Zhao$^*$,~ Xinyue Li} \\
		[1em]
		\footnotesize College of Mathematics and Systems Science, Shandong University of Science and Technology,\\
		\footnotesize Qingdao, 266590, Shandong, P. R. China. \\
		}

	\maketitle
		\begin{abstract}
	We construct a local tri-Hamiltonian structure 
	$(\mathcal{P}_1,\mathcal{P}_2,\mathcal{P}_3)$ of the asymmetric $(3,1)$-type generalized Ablowitz--Ladik (gAL) hierarchy at the full-dispersive level and rigorously prove its validity using the supervariable technique. All central invariants of the corresponding bi-Hamiltonian structures are computed, with the notable result that those of $(\mathcal{P}_1,\mathcal{P}_2)$ are equal to $\frac{1}{24}$. In addition, we construct a Frobenius manifold $M$ arising from the dispersionless limit of this hierarchy and show that the dispersionless limits of the first flows of the $(3,1)$-type gAL hierarchy belong to the Principal Hierarchy of $M$.
	\end{abstract}
	\\[6pt]\textbf{MSC(2020)} 37K10, 37K25, 53D45
	\tableofcontents
	\section{Introduction}
The Ablowitz--Ladik (AL) equation \cite{AL1975,AL1976}
\begin{equation}\label{AL}
	i\partial_t q_n = q_{n+1} - 2q_n + q_{n-1} \pm q_n^* q_n (q_{n+1} + q_{n-1}), \quad n \in \mathbb{Z}
\end{equation}
is one of the most paradigmatic nonlinear differential-difference equations in the theory of integrable systems, serving as an integrable discretization of the nonlinear Schr\"odinger equation. It has attracted considerable attention \cite{Fan,Geng,AL symmetries1} and is closely related to several classical integrable systems \cite{Universality}, including the two-dimensional Toda (2D-Toda) lattice \cite{AL&2DT}, the Volterra lattice \cite{Suris}, and the relativistic Toda lattice \cite{relativistic Toda2,relativistic Toda1}. The AL hierarchy \cite{Suris} consists of infinitely many differential-difference equations commuting with \eqref{AL}, playing an important role in the description of the Gromov--Witten invariants of local $\mathbb{CP}^1$ \cite{Brini CMP,Brini 2012}. This hierarchy possesses a local tri-Hamiltonian structure \cite{AL-triham}, and its dispersionless limit corresponds to a generalized Frobenius manifold $M_{AL}$ with a non-flat unity \cite{Liu2024,GFM}. The extended AL hierarchy further confirms the consistency of the AL hierarchy with the topological deformation of the Principal Hierarchy of $M_{AL}$ \cite{extend AL}.

Inspired by Brini's conjecture \cite{Brini CMP,Brini 2012}, Takasaki found that the associated Lax operators of the Toda hierarchy take a particular factorized form, which coincides with the condition characterizing the AL hierarchy as a reduction of the Toda hierarchy. In 2014, he identified the fundamental integrable structure of topological string theory on generalized conifolds as the generalized Ablowitz--Ladik (gAL) hierarchy \cite{Takasaki GAL}, whose Lax operators admit several  equivalent factorized forms. One such form is given by
\begin{equation}\label{gAL}
	L_{(N,N)}=C^{-1}B \varLambda^{1-N}, \quad N\in\mathbb{Z}_{>0},
\end{equation}
where $\varLambda = \mathrm{e}^{\varepsilon \partial_x}$ is the shift operator acting as $\varLambda f(x) = f(x+\varepsilon)$ for any $f(x)$, $\varepsilon$ is a dispersion parameter, and
\begin{equation*}
	B=\varLambda^{N}+\alpha_{1}\varLambda^{N-1}+\cdots+\alpha_{N},\quad
	C=1+\beta_{1}\varLambda^{-1}+\cdots+\beta_{N}\varLambda^{-N}
\end{equation*}
with unknown functions $\alpha_1,\dots,\alpha_N$ and $\beta_1,\dots,\beta_N$. The factorized form~\eqref{gAL} is preserved by the flows of the Toda hierarchy; see \cite{Takasaki GAL} for more details.
	
Motivated by Takasaki's work, we generalize \eqref{gAL} to the following form 
\begin{equation}\label{ab type gAL}
	L_{(a,b)} = \left( 1 + \beta_1 \varLambda^{-1} + \cdots + \beta_b \varLambda^{-b}\right)^{-1}
	\left( \varLambda + \alpha_{1} +\alpha_2\varLambda^{-1} +\cdots+\alpha_{a}\varLambda^{1-a} \right),
\end{equation}
where $a,b \in \mathbb{Z}_{>0}$. When $a=b$, we call it the symmetric gAL hierarchy, which is equivalent to \eqref{gAL}. In the special case $a=b=1$, it reduces further to the AL hierarchy. When $a\neq b$, we call it the asymmetric $(a,b)$-type gAL hierarchy. This paper focuses on the asymmetric $(3,1)$-type gAL hierarchy, whose Lax operator reads
\begin{equation}\label{Lax-operator}	L=(1+w\varLambda^{-1})^{-1}(\varLambda+u_{1}+u_{2}\varLambda^{-1}+u_{3}\varLambda^{-2}).
\end{equation}
We explicitly derive three Hamiltonian operators \(\mathcal{P}_1\), \(\mathcal{P}_2\) and \(\mathcal{P}_3\) for this hierarchy, and show that they constitute a local tri-Hamiltonian structure by the supervariable technique~\cite{Kersten}. All central invariants of the resulting bi-Hamiltonian structures $(\mathcal P_{i},\mathcal P_{j})$ ($i,j=1,2,3$, $i\neq j$) are computed; in particular, those of the pair $(\mathcal P_{1},\mathcal P_{2})$ are equal to $\frac{1}{24}$. On the geometric side, we consider the dispersionless limit of this hierarchy and construct an associated Frobenius manifold $M$.  We then show that the dispersionless limits of the  first flows of the \((3,1)\)-type gAL hierarchy are contained in the Principal Hierarchy of $M$. Finally, we present three explicit Legendre transformations of $M$.

This paper is organized as follows. In Sect.~2, we introduce the Lax representation and Hamiltonian formalism of the $(3,1)$-type gAL hierarchy. Sect.~3 is devoted to the proof of its local tri-Hamiltonian structure and the computation of the associated central invariants. In Sect.~4, we construct a Frobenius manifold arising from the dispersionless limit of this hierarchy. Sect.~5 summarizes our conclusions and outlines future directions.
		\section{The asymmetric $(3,1)$-type gAL hierarchy }
 In this section, we present Lax representation of $(3,1)$-type gAL hierarchy and derive its Hamiltonian formalism.
	\subsection{Lax representation of the $(3,1)$-type gAL hierarchy}
	For convenience, we introduce the operators corresponding to 	\eqref{Lax-operator} 
	\begin{equation}
		X=\varLambda+u_{1}+u_{2}\varLambda^{-1}+u_{3}\varLambda^{-2},~~Y=1+w\varLambda^{-1},
	\end{equation}
 and denote
	\begin{align}
		L=Y^{-1}X,\quad \widehat{L}=XY^{-1}.
	\end{align}
	
	To facilitate the definition of the positive flows of the $(3,1)$-type gAL hierarchy, we expand the operators $L$ and $\widehat{L}$ as 
		\begin{align}
		&L=\varLambda+(u_{1}-w)+(u_{2}-wu_{1}^{-}+ww^{-})\varLambda^{-1}+\cdots,\label{L1}\\[6pt]
		&\widehat{L}=\varLambda+(u_{1}-w^{+})+(u_{2}-u_{1}w+ww^{+})\varLambda^{-1}+\cdots,\label{Lh1}
	\end{align}
	where $f^\pm = \varLambda^{\pm1} f$, $f^{[i]} = \varLambda^i (f)$ for all $i \in \mathbb{Z}$.
Correspondingly, for the negative flows of the hierarchy, we expand $L$ and $\widehat{L}$ as
	\begin{align}
			&L=\,\frac{u_{3}^{+}}{w^{+}}\varLambda^{-1}+\left(\frac{u_{2}^{+}}{w^{+}}-\frac{u_{3}^{++}}{w^{+}w^{++}}\right)+\left(\frac{u_{1}^{+}}{w^{+}}-\frac{u_{2}^{++}}{w^{+}w^{++}}+\frac{u_{3}^{+++}}{w^{+}w^{++}w^{+++}}\right)\varLambda+\cdots,\label{L2} \\[6pt]
	&\widehat{L}=\,\frac{u_{3}}{w^{-}}\varLambda^{-1}+\left(\frac{u_{2}}{w}-\frac{u_{3}}{ww^{-}}\right)+\left(\frac{u_{1}}{w^{+}}-\frac{u_{2}}{ww^{+}}+\frac{u_{3}}{w^{-}ww^{+}}\right)\varLambda+\cdots.\label{Lh2}
	\end{align}
		For a formal Laurent difference operator $A=\sum_{j\in\mathbb{Z}}a_j\varLambda^j$, we define its positive and negative parts by
	\begin{equation*}
		A_+ = \sum_{j\geq 0} a_j \varLambda^j, \qquad A_- = \sum_{j< 0} a_j \varLambda^j,
	\end{equation*}
	so that $A = A_+ + A_-$.
	\begin{defn}
		The $(3,1)$-type gAL hierarchy consists of the following differential–difference equations:
			\begin{align}
			\veps\frac{\partial X}{\partial t^{4,m}}&=\,\frac{1}{(m+1)!}\left(\left(\widehat{L}^{m+1}\right)_{+}X-X\left(L^{m+1}\right)_{+}\right),\label{Lax eq1}\\ 
			\veps\frac{\partial Y}{\partial t^{4,m}}&=\,\frac{1}{(m+1)!}\left(\left(\widehat{L}^{m+1}\right)_{+}Y-Y\left(L^{m+1}\right)_{+}\right),\label{Lax eq2}\\
			\veps\frac{\partial X}{\partial  t^{2,m}}&=\,-\frac{1}{(m+1)!}\left(\left(\widehat{L}^{m+1}\right)_{-}X-X\left(L^{m+1}\right)_{-}\right),\label{Lax eq3}\\
			\veps\frac{\partial Y}{\partial  t^{2,m}}&=\,-\frac{1}{(m+1)!}\left(\left(\widehat{L}^{m+1}\right)_{-}Y-Y\left(L^{m+1}\right)_{-}\right)  \label{Lax eq4}
		\end{align}
for all $m \geq 0$.
Here, $\{t^{2,m}, t^{4,m}\}_{m\geq 0}$ denotes the family of time variables, and we refer to the flows $\frac{\partial}{\partial t^{2,m}}$ and $\frac{\partial}{\partial t^{4,m}}$ as the negative  and positive flows of the hierarchy, respectively. In particular, identifying the spatial variable $x$ with $t^{0,0}$, we have
\begin{equation}
	\frac{\partial}{\partial t^{0,0}} = \frac{\partial}{\partial x},
\end{equation}
and we write $t^{0,0}$ as $x$ in what follows.
\end{defn}
	
By direct calculation, equations \eqref{Lax eq1}--\eqref{Lax eq4} are equivalent to
	\begin{align}
		\veps\frac{\partial  L}{\partial t^{4,m}}
		&=\,\frac{1}{(m+1)!}\left[\left(L^{m+1}\right) _+, L\right],\quad m\geq 0,
		\\
		\veps\frac{\partial  L}{\partial t^{2,m}}
		&=\,
		-\frac{1}{(m+1)!}\left[\left(L^{m+1}\right)_-,L\right], \quad m\geq 0,
	\end{align}
	which are precisely the Lax equations of the $(3,1)$-type gAL hierarchy.

	Using \eqref{L1} and \eqref{Lh1}, we write the Lax operators $L$ and $\widehat{L}$ as
	\begin{align}\label{LmLhm1}
		L^{m}&=\sum_{r\geq 0}a^{m}_{r}\varLambda ^{m-r},
	\qquad
	\widehat{L}^{m}=\sum_{r\geq 0}b^{m}_{r}\varLambda ^{m-r},
	\end{align}
from which the positive flows of the hierarchy are generated by
	\begin{align}
		\veps\pfrac{u_1}{t^{4,m}}
		&= \frac{1}{(m+1)!} \biggl(
		u_{1}\bigl(b_{m+1}^{m+1}-a_{m+1}^{m+1}\bigr)
		+ u_{2}^{\vphantom{+}+}b_{m}^{m+1}
		- u_{2}\bigl(a_{m}^{m+1}\bigr)^{\vphantom{+}-} + u_{3}^{\vphantom{++}++}b_{m-1}^{m+1}
		\notag \\
		&\phantom{= \frac{1}{(m+1)!} \biggl(}
		- u_{3}\bigl(a_{m-1}^{m+1}\bigr)^{\vphantom{--}--}
		\biggr),
		\label{Du1t4m} \\
		\veps\pfrac{u_2}{t^{4,m}}
		&= \frac{1}{(m+1)!} \biggl(
		u_{2}b_{m+1}^{m+1}
		- u_{2}\bigl(a_{m+1}^{m+1}\bigr)^{\vphantom{+}-}
		+ u_{3}^{\vphantom{+}+}b_{m}^{m+1}
		- u_{3}\bigl(a_{m}^{m+1}\bigr)^{\vphantom{--}--}
		\biggr),
		\label{Du2t4m} \\
		\veps\pfrac{u_3}{t^{4,m}}
		&= \frac{1}{(m+1)!} \biggl(
		u_{3}b_{m+1}^{m+1}
		- u_{3}\bigl(a_{m+1}^{m+1}\bigr)^{\vphantom{--}--}
		\biggr),
		\label{Du3t4m} \\
		\veps\pfrac{w}{t^{4,m}}
		&= \frac{1}{(m+1)!} \biggl(
		wb_{m+1}^{m+1}
		- w\bigl(a_{m+1}^{m+1}\bigr)^{\vphantom{+}-}
		\biggr).
		\label{Dwt4m}
	\end{align}
Correspondingly, using \eqref{L2} and \eqref{Lh2}, we have
		\begin{align}\label{LmLhm2}
		L^{m}=\sum_{r\geq 0}c^{m}_{r}\varLambda ^{-m+r},\qquad
		\widehat{L}^{m}=\sum_{r\geq 0}d^{m}_{r}\varLambda ^{-m+r},
	\end{align}
from which the negative flows of the hierarchy take the following form
	\begin{align}
		\veps\pfrac{u_1}{t^{2,m}}
		&= -\frac{1}{(m+1)!} \biggl(
		d^{m+1}_{m} - (c^{m+1}_{m})^{\vphantom{+}+}
		\biggr),
		\label{Du1t2m} \\
		\veps\pfrac{u_2}{t^{2,m}}
		&= -\frac{1}{(m+1)!} \biggl(
		u_{1}^{\vphantom{+}-}d_{m}^{m+1} - u_{1}c^{m+1}_{m}
		+ d^{m+1}_{m-1} - (c^{m+1}_{m-1})^{\vphantom{+}+}
		\biggr),
		\label{Du2t2m} \\
		\veps\pfrac{u_3}{t^{2,m}}
		&= -\frac{1}{(m+1)!} \biggl(
		u_{1}^{\vphantom{--}--}d^{m+1}_{m-1} - u_{1}c^{m+1}_{m-1}
		+ u_{2}^{\vphantom{+}-}d^{m+1}_{m} - u_{2}(c^{m+1}_{m})^{\vphantom{+}-}+ d^{m+1}_{m-2}
		\notag \\
		&\phantom{= -\frac{1}{(m+1)!} \biggl(}
		 - (c^{m+1}_{m-2})^{\vphantom{+}+}
		\biggr),
		\label{Du3t2m} \\
		\veps\pfrac{w}{t^{2,m}}
		&= -\frac{1}{(m+1)!} \biggl(
		d_{m}^{m+1} - c^{m+1}_{m}
		\biggr).
		\label{Dwt2m}
	\end{align}
	Using the identity
	\begin{align}\label{XYLLh}
		XL^{m}&=\widehat{L}^{m}X=\widehat{L}^{m+1}Y=YL^{m+1},
	\end{align}
	the recurrence relations for the coefficients $a_{r}^{m}$, $b_{r}^{m}$ and $c_{r}^{m}$, $d_{r}^{m}$ are obtained via
	\begin{equation}\label{abrm}
		\begin{aligned}
			&(a_{r+1}^{m})^{+} + u_{1}a_{r}^{m} + u_{2}(a_{r-1}^{m})^{-} + u_{3}(a_{r-2}^{m})^{--} \\
			={}&b_{r+1}^{m} + b_{r}^{m}\varLambda^{m-r}(u_{1}) + b_{r-1}^{m}\varLambda^{m+1-r}(u_{2}) + b_{r-2}^{m}\varLambda^{m+2-r}(u_{3})~~~~~~~ \\
			={}&b_{r+1}^{m+1} + b_{r}^{m+1}\varLambda^{m+1-r}(w) \\
			={}&a_{r+1}^{m+1} + w(a_{r}^{m+1})^{-},
		\end{aligned}
	\end{equation}
	and
	\begin{equation}\label{cdrm}
		\begin{aligned}
			&(c_{r-1}^{m})^{+} + u_{1}c_{r}^{m} + u_{2}(c_{r+1}^{m})^{-} + u_{3}(c_{r+2}^{m})^{--} \\
			={}&d_{r-1}^{m} + d_{r}^{m}\varLambda^{-m+r}(u_{1}) + d_{r+1}^{m}\varLambda^{-m+1+r}(u_{2}) + d_{r+2}^{m}\varLambda^{-m+2+r}(u_{3}) \\
			={}&d_{r+1}^{m+1} + d_{r+2}^{m+1}\varLambda^{-m+1+r}(w) \\
			={}&c_{r+1}^{m+1} + w(c_{r+2}^{m+1})^{-}.
		\end{aligned}
	\end{equation}
With the initial conditions	
	\begin{align*}
		a_0^m &= b_0^m = 1, \quad a_r^0 = b_r^0 = 0, \quad m \geq 0,\ r \geq 1, \\
		c_0^0 &= d_0^0 = 1, \quad c_r^0 = d_r^0 = 0, \quad r \geq 1, \\
		c_0^m &= \prod_{i=0}^{m-1} \varLambda^{-i}\left( \frac{u_3^{+}}{w^{+}} \right), \quad d_0^m = \prod_{i=0}^{m-1} \varLambda^{-i}\left( \frac{u_3}{w^{-}} \right), \quad m \geq 1,
	\end{align*}
	all coefficients $a_r^m, b_r^m$ and $c_r^m, d_r^m$ are uniquely determined by the \eqref{abrm} and \eqref{cdrm}. The first examples read
	\begin{equation*}
		a_1^1 = u_1-w,\qquad b_1^1 = u_1-w^+, 
	\end{equation*}
	and
	\begin{equation*}
		c_1^1 = \frac{u_{2}^{+}}{w^+} - \frac{u_3^{++}}{w^+w^{++}},\qquad d_1^1 = \frac{u_2}{w} - \frac{u_3}{ww^-}.
	\end{equation*}
	
	\begin{ex}
		Using \eqref{Du1t4m}--\eqref{Dwt2m} and by straightforward calculation, we obtain the first positive flows of the $(3,1)$-type gAL hierarchy
		\begin{alignat*}{2}
			\veps\pfrac{u_1}{t^{4,0}}&=\,
			u_{1}(w-w^{+})+u^{+}_{2}-u_{2},\quad &
			\veps\pfrac{u_2}{t^{4,0}}&=\,	u_{2}(u_{1}-u_{1}^{-}+w^{-}-w^{+})+u^{+}_{3}-u_{3},\\
			\veps\pfrac{u_3}{t^{4,0}}&=\,u_{3}(u_{1}-u_{1}^{--}+w^{--}-w^{+}),\quad &
			\veps\pfrac{w}{t^{4,0}}
			&=\,w(u_{1}-u_{1}^{-}+w^{-}-w^{+}),
		\end{alignat*}
		along with the corresponding first negative flows
\begin{alignat*}{2}
	\varepsilon \pfrac{u_1}{t^{2,0}} &= \frac{u_3^{++}}{w^{++}} - \frac{u_3}{w^-}, \qquad &
	\varepsilon \pfrac{u_2}{t^{2,0}} &= \frac{u_1 u_3^+}{w^+} - \frac{u_1^- u_3}{w^-}, \\[4pt]
	\varepsilon \pfrac{u_3}{t^{2,0}} &= \frac{u_2 u_3}{w} - \frac{u_2^- u_3}{w^-}, \qquad &
	\varepsilon \pfrac{w}{t^{2,0}} &= \frac{u_3^+}{w^+} - \frac{u_3}{w^-}.
\end{alignat*}
	\end{ex}

			\subsection{Hamiltonian formalism of the $(3,1)$-type gAL hierarchy}
To derive the Hamiltonian formalism of the $(3,1)$-type gAL hierarchy, we begin by defining the corresponding Hamiltonians
\begin{equation*}
	H_{4,m} = \int h_{4,m+1} \, dx, \quad H_{2,m} = \int h_{2,m+1} \, dx, \quad m \geq -1,
\end{equation*}
where the Hamiltonian densities are defined as
\begin{align*}
	h_{4,m} = \frac{1}{(m+2)!} \operatorname{Res} L^{m+2}, \quad 
	h_{2,m} = -\frac{1}{(m+2)!} \operatorname{Res} L^{m+2}, \quad m \geq 0.
\end{align*}
	Here, we define the residue of an operator $A=\sum_{j\in\mathbb{Z}}a_j\varLambda^j$ by $\operatorname{Res} A = a_0$, and introduce the variational derivatives of a functional $H = \int h\,\mathrm{d}x$ as
	\begin{equation*}
		\frac{\delta H}{\delta u_i(n)} = \sum_{r\in\mathbb{Z}} \varLambda^{-r} \frac{\partial h(n)}{\partial u_i(n+r)},
		\quad
		\frac{\delta H}{\delta w(n)} = \sum_{r\in\mathbb{Z}} \varLambda^{-r} \frac{\partial h(n)}{\partial w(n+r)},\quad i=1,2,3.
	\end{equation*}
 In what follows, we restrict the Hamiltonian formalism to the $\pp{t^{4,m}}$-flows; the derivation for the $\pp{t^{2,m}}$-flows is completely analogous and thus omitted. We first introduce some lemmas.
\begin{lem}\label{var-derivative}
	For every integer $m \geq 0$, the variational derivatives of the Hamiltonian $H_{4,m}$ with respect to the functions $u_1, u_2, u_3, w$ satisfy the following identities:
\begin{alignat*}{2}
	\frac{\delta H_{4,m}}{\delta u_1} &= \frac{1}{(m+1)!} k_{m+1}^{m+1}, \qquad &
	\frac{\delta H_{4,m}}{\delta u_2} &= \frac{1}{(m+1)!} (k_{m}^{m+1})^-, \\[4pt]
	\frac{\delta H_{4,m}}{\delta u_3} &= \frac{1}{(m+1)!} (k_{m-1}^{m+1})^{--}, \qquad &
	\frac{\delta H_{4,m}}{\delta w} &= -\frac{1}{(m+1)!} (k_{m+1}^{m+2})^-
\end{alignat*}
		with the coefficients $k_r^m$ given by
		\begin{equation}\label{kmr}
			L^m Y^{-1} = \sum_{r\in\mathbb{Z}} k_r^m \varLambda^{m-r}.
		\end{equation}
\end{lem}
\begin{proof}
	Let us first define the space of variational 1-forms by 
	\begin{equation*}
		\omega = \sum_{r\in \mathbb{Z}} \left( f_r \delta u_1^{(r)} + g_r \delta u_2^{(r)} + p_r \delta u_3^{(r)} + q_r \delta w^{(r)} \right),
	\end{equation*}
	where $f_r$, $g_r$, $p_r$ and $q_r$ are smooth functions of $u_i^{(k)}(n) = u_i(n+k)$ for $i=1,2,3$ and $w^{(k)}(n) = w(n+k)$ $(k\in \mathbb{Z})$. We next equip this space with an equivalence relation $\sim$: two 1-forms $\omega_1$ and $\omega_2$ are equivalent if and only if 
 $\omega_1 - \omega_2 = (\varLambda - 1)\omega'$
	for some variational 1-form $\omega'$. Under this equivalence, one has
		\begin{align*}
			\delta h_{4,m} &=\, \sum_{r\in \mathbb{Z}} \left( \frac{\partial h_{4,m}}{\partial u_1^{(r)}} \delta u_1^{(r)} + \frac{\partial h_{4,m}}{\partial u_2^{(r)}} \delta u_2^{(r)} + \frac{\partial h_{4,m}}{\partial u_3^{(r)}} \delta u_3^{(r)} + \frac{\partial h_{4,m}}{\partial w^{(r)}} \delta w^{(r)} \right) \\[2pt]		&\sim\, \frac{\delta H_{4,m}}{\delta u_1} \delta u_1 + \frac{\delta H_{4,m}}{\delta u_2} \delta u_2 + \frac{\delta H_{4,m}}{\delta u_3} \delta u_3 + \frac{\delta H_{4,m}}{\delta w} \delta w.
		\end{align*}
	Moreover, from the definition of $h_{4,m}$, we deduce
	\begin{align*}
		\delta h_{4,m} &\sim \frac{1}{(m+1)!} \operatorname{Res} \left( L^{m+1} \delta L \right)\\&=\frac{1}{(m+1)!} \operatorname{Res} \left( L^{m+1} Y^{-1} \left( \delta u_1 + \delta u_2 \varLambda^{-1} + \delta u_3 \varLambda^{-2} - \delta w \varLambda^{-1} L \right) \right) \\
		&\sim \frac{1}{(m+1)!} \operatorname{Res} \big( \delta u_1 L^{m+1} Y^{-1} + \delta u_2 \varLambda^{-1} L^{m+1} Y^{-1} + \delta u_3 \varLambda^{-2} L^{m+1} Y^{-1} \\
		&\qquad \qquad \qquad \quad- \delta w \varLambda^{-1} L^{m+2} Y^{-1} \big).
	\end{align*}
	Combining the above expressions and matching coefficients, we arrive at the variational derivative relation
	\begin{align*}
		\frac{\delta H_{4,m}}{\delta u_1} &=\frac{1}{(m+1)!} \operatorname{Res} \left( L^{m+1} Y^{-1} \right) =\frac{1}{(m+1)!} k_{m+1}^{m+1}.
	\end{align*}
	We can derive $\frac{\delta H_{4,m}}{\delta u_2},\frac{\delta H_{4,m}}{\delta u_3}$ and $\frac{\delta H_{4,m}}{\delta w}$ in a similar way. This completes the proof of the Lemma~\ref{var-derivative}.	
\end{proof}
\begin{lem}\label{abk}
	For every integer $m \geq 0$, the following relations hold true:
\begin{equation}
	\begin{aligned}
		a_{r}^{m}&=k_{r}^{m}+k_{r-1}^{m} w^{[m+1-r]},\\ b_{r}^{m}&=k_{r}^{m}+w(k_{r-1}^{m})^{-},\\
		a_{r}^{m+1}&=k_{r}^{m}+k_{r-1}^{m}u_{1}^{[m+1-r]}+k_{r-2}^{m}u_{2}^{[m+2-r]}+k_{r-3}^{m}u_{3}^{[m+3-r]},\\
		b_{r}^{m+1}&=(k_{r}^{m})^{+}+u_{1}k_{r-1}^{m}+u_{2}(k_{r-2}^{m})^{-}+u_{3}(k_{r-3}^{m})^{--}.
	\end{aligned}
\end{equation}
\end{lem}
\noindent These relations follow directly from the notations \eqref{LmLhm1}, \eqref{kmr} and the identity \eqref{XYLLh}.

For later convenience, we introduce the difference operators
\begin{alignat*}{3}
	\varDelta &= \varLambda-1, \quad &	\varGamma &= \varLambda+1, \quad &
	\varPhi &= \varLambda^2+\varLambda+1,
	\\
	\varTheta &=\varLambda^{-1}-1,\quad &	\varOmega &= \varLambda^{-1}+1, \quad &
	\varPsi &= \varLambda^{-2}+\varLambda^{-1}+1.
\end{alignat*}
\begin{thm}\label{thm:HamFormalism-1}
	The $(3,1)$-type gAL hierarchy \eqref{Lax eq1}--\eqref{Lax eq4}
	can be represented by the following Hamiltonian formalism:
	\begin{equation*}
		\veps\left(\pfrac{u_1}{t^{i,m}}, \pfrac{u_2}{t^{i,m}}, \pfrac{u_3}{t^{i,m}}, \pfrac{w}{t^{i,m}}\right)^\mathrm{T}=\mathcal P_{1}\left(\frac{\delta{H_{i,m}}}{\delta{u_{1}}},\frac{\delta{H_{i,m}}}{\delta{u_{2}}},\frac{\delta{H_{i,m}}}{\delta{u_{3}}},\frac{\delta{H_{i,m}}}{\delta{w}} \right)^\mathrm{T}
	\end{equation*}
	for $i\in\{2,4\}$, $m\geq 0$, where
	\begin{equation}\label{HamP1}
		\mathcal P_{1}=
		\begin{pmatrix}
			\varLambda w-w\varLambda^{-1}& \varDelta u_{2} &\varDelta\varGamma u_{3} &\varDelta w\\[6pt]
			-u_{2}\varTheta&\varLambda u_{3}-u_{3}\varLambda^{-1}& w\varDelta u_{3} &0\\[6pt]
			-u_{3}\varTheta\varOmega&-u_{3}\varTheta w & 0 &0\\[6pt]
			-w\varTheta&0 & 0 &0
		\end{pmatrix}.
	\end{equation}
\end{thm}
\begin{proof}
	Using Lemmas \ref{var-derivative} and \ref{abk},
	equations \eqref{Du1t4m}--\eqref{Dwt4m} can be represented as follows:
	\begin{align*}
		\veps \pfrac{u_1}{t^{4,m}}
		&= \frac{1}{(m+1)!} \biggl(\Bigl(w\varLambda^{2}u_{3}-u_{3}\varLambda w\Bigr)(k^{m+1}_{m-2})^{---}+
		\Bigl(\varDelta u_{2}-u_{1}\varDelta w\Bigr)(k^{m+1}_{m})^{-} \nonumber\\
		&\qquad \qquad \qquad + \Bigl(w\varLambda u_{2}-u_{2}\varLambda w+\varDelta\varGamma u_{3}\Bigr)(k^{m+1}_{m-1})^{--}
		\biggr),\nonumber\\
		\veps \pfrac{u_2}{t^{4,m}}
		&= \frac{1}{(m+1)!} \biggl(
		-u_{2}\varTheta k^{m+1}_{m+1}
		+ \bigl(\varLambda u_{3}-u_{3}\varLambda^{-1}\bigr)(k^{m+1}_{m})^{-}
		+ w\varDelta u_3(k^{m+1}_{m-1})^{--}
		\biggr) \nonumber\\
		&= -u_{2}\varTheta\frac{\delta H_{4,m}}{\delta u_1}
		+ \bigl(\varLambda u_{3}-u_{3}\varLambda^{-1}\bigr)\frac{\delta H_{4,m}}{\delta u_2}
		+ w\varDelta u_3\frac{\delta H_{4,m}}{\delta u_3}, \nonumber\\[6pt]
		\veps \pfrac{u_3}{t^{4,m}}
		&= \frac{1}{(m+1)!} \biggl(
		-u_{3}\varTheta\varOmega k^{m+1}_{m+1}
		- u_{3}\varTheta w(k^{m+1}_{m})^{-}
		\biggr)\\& = -u_{3}\varTheta\varOmega\frac{\delta H_{4,m}}{\delta u_1}
		- u_{3}\varTheta w\frac{\delta H_{4,m}}{\delta u_2}, \nonumber\\[6pt]
		\veps \pfrac{w}{t^{4,m}}
		&= \frac{1}{(m+1)!} \biggl( -w\varTheta k^{m+1}_{m+1} \biggr) = -w\varTheta \frac{\delta H_{4,m}}{\delta u_1}.
	\end{align*}
	Using Lemma \ref{abk}, we consider the coefficients $a^{m+2}_{m+1}$ and $b^{m+2}_{m+1}$, from which we obtain
	\begin{align*}
		\left(w\varLambda^{2}u_{3}-u_{3}\varLambda w\right)(k^{m+1}_{m-2})^{---}=&\left(\varLambda w-w\varLambda^{-1}\right)k^{m+1}_{m+1}+u_{1}\varDelta w(k^{m+1}_{m})^{-}\\&+\left(u_2\varLambda w-w\varLambda u_2\right)(k^{m+1}_{m-1})^{--}-\varDelta w(k^{m+2}_{m+1})^{-}.
	\end{align*}
	Substituting this identity into the $\veps \pfrac{u_1}{t^{4,m}}$-flow, this completes the proof of Theorem~\ref{thm:HamFormalism-1}.
\end{proof}

\begin{thm}\label{thm:HamFormalism-2}
	The $(3,1)$-type gAL hierarchy \eqref{Lax eq1}--\eqref{Lax eq4}
	can also be represented by the following Hamiltonian formalism:
	\begin{equation*}
		\veps\left(\pfrac{u_1}{t^{i,m+1}}, \pfrac{u_2}{t^{i,m+1}}, \pfrac{u_3}{t^{i,m+1}}, \pfrac{w}{t^{i,m+1}}\right)^\mathrm{T}=c_2^{i,m}\mathcal P_{2}\left(\frac{\delta{H_{i,m}}}{\delta{u_{1}}},\frac{\delta{H_{i,m}}}{\delta{u_{2}}},\frac{\delta{H_{i,m}}}{\delta{u_{3}}},\frac{\delta{H_{i,m}}}{\delta{w}} \right)^\mathrm{T}
	\end{equation*}
	with $c_2^{i,m}=\frac{1}{m+2}$ for $i=2,4$ and $m\geq -1$. The skew-adjoint
	operator-valued matrix $\mathcal{P}_2$ takes the form
	\begin{equation}\label{HamP2}
		\mathcal{P}_{2}=
		\begin{pmatrix}
			\varLambda u_{2} - u_{2}\varLambda^{-1}
			& \mathcal{P}_2^{12} & u_{1}\varDelta\varGamma u_{3} & u_{1}\varDelta w\\[8pt]
			-\left(\mathcal{P}_2^{12}\right)^\dagger &\mathcal{P}_2^{22} & u_{2}\varDelta \varOmega\varGamma u_{3} & u_{2}\varDelta \varOmega w\\[8pt]
			-u_{3}\varTheta\varOmega u_{1} & u_{3}\varDelta \varOmega^2 u_{2} & u_{3}\varDelta \varPsi \varGamma u_{3} & u_{3}\varDelta \varPsi w\\[8pt]
			-w\varTheta u_{1} & w\varDelta \varOmega u_{2} & w\varTheta \varPhi u_{3} & w\varDelta \varOmega w
		\end{pmatrix},
	\end{equation}
where \begin{equation*}
	\mathcal{P}_2^{12} = \varLambda^{2} u_{3} - u_{3}\varLambda^{-1} + u_{1}\varDelta u_{2},\quad 
	\mathcal{P}_2^{22} = u_{1}\varLambda u_{3} - u_{3}\varLambda^{-1}u_{1} + u_{2} \varDelta \varOmega u_{2}.
\end{equation*}
\end{thm}
\begin{proof}
	Using Lemmas \ref{var-derivative} and \ref{abk}, we obtain 
	\begin{align*}
		\veps \pfrac{u_1}{t^{4,m+1}}
		&=\frac{1}{(m+2)!}\bigg((\varLambda u_{2}-u_{2}\varLambda^{-1})k^{m+1}_{m+1}+(\varLambda^{2}u_{3}-u_{3}\varLambda^{-1})(k^{m+1}_{m})^{-}+u_{1}\varDelta k^{m+1}_{m+2}\bigg),\nonumber\\
		\veps \pfrac{u_2}{t^{4,m+1}}
		&=\frac{1}{(m+2)!}\bigg((\varLambda u_{3}-u_{3}\varLambda^{-2}-u_{2}\varTheta u_{1})k^{m+1}_{m+1}+(u_{1}\varLambda u_{3}-u_{3}\varLambda^{-1}u_{1})(k^{m+1}_{m})^{-}\nonumber\\	&\qquad \qquad \qquad+u_{2}\varDelta \varOmega k^{m+1}_{m+2}\bigg),\nonumber\\
		\veps \pfrac{u_3}{t^{4,m+1}}
		&=\frac{1}{(m+2)!}\bigg(-u_{3}\varTheta\varOmega u_{1}k^{m+1}_{m+1}-u_{3}\varTheta u_{2}(k^{m+1}_{m})^{-}+u_{3}\varDelta \varPsi k^{m+1}_{m+2}\bigg),\nonumber\\
		\veps \pfrac{w}{t^{4,m+1}}
		&=\frac{1}{(m+2)!}\bigg(-w\varTheta u_{1}k^{m+1}_{m+1}-w\varDelta u_{3}(k^{m+1}_{m-1})^{--}+w\varDelta \varOmega k^{m+1}_{m+2}\bigg).
	\end{align*}
	Once $k^{m+1}_{m+2}$ is expressed in terms of $k^{m+1}_{m+1}$, $(k^{m+1}_{m})^{-}$, $(k^{m+1}_{m-1})^{--}$ and $(k^{m+2}_{m+1})^{-}$, the proof of Theorem \ref{thm:HamFormalism-2} follows via straightforward computations.
	
	Using Lemma \ref{abk}, we consider the coefficients $a^{m+2}_{m+2}$ and $b^{m+2}_{m+2}$, from which we obtain
	\begin{align*}
		k^{m+1}_{m+2}&=u_{2}(k^{m+1}_{m})^{-}+\varGamma u_{3}(k^{m+1}_{m-1})^{--}-w(k^{m+2}_{m+1})^{-}.
	\end{align*}
	Substituting this identity into the $\pp{t^{4,m+1}}$-flows, we complete the proof of Theorem \ref{thm:HamFormalism-2}.
\end{proof}
\begin{thm}\label{thm:HamFormalism-3}
	The $(3,1)$-type gAL hierarchy \eqref{Lax eq1}--\eqref{Lax eq4}
	also admits the following Hamiltonian formalism:
	\begin{equation*}
		\veps\left(\pfrac{u_1}{t^{i,m+2}}, \pfrac{u_2}{t^{i,m+2}}, \pfrac{u_3}{t^{i,m+2}}, \pfrac{w}{t^{i,m+2}}\right)^\mathrm{T}=c_3^{i,m}\mathcal P_{3}\left(\frac{\delta{H_{i,m}}}{\delta{u_{1}}},\frac{\delta{H_{i,m}}}{\delta{u_{2}}},\frac{\delta{H_{i,m}}}{\delta{u_{3}}},\frac{\delta{H_{i,m}}}{\delta{w}} \right)^\mathrm{T}
	\end{equation*}
	with $c_3^{i,m} = \frac{1}{(m+2)(m+3)}$ for $i=2,4$ and $m \geq -1$,
	and
	\begin{equation*}
		\veps\left(\pfrac{u_1}{t^{i,0}}, \pfrac{u_2}{t^{i,0}}, \pfrac{u_3}{t^{i,0}}, \pfrac{w}{t^{i,0}}\right)^\mathrm{T}=\mathcal P_{3}\left(\frac{\delta{H_{i,-2}}}{\delta{u_{1}}},\frac{\delta{H_{i,-2}}}{\delta{u_{2}}},\frac{\delta{H_{i,-2}}}{\delta{u_{3}}},\frac{\delta{H_{i,-2}}}{\delta{w}} \right)^\mathrm{T},\quad i=2,4.
	\end{equation*}
	
	The skew-adjoint operator-valued matrix \begin{equation}\label{HamP3}
		\mathcal{P}_3=\left(\mathcal{P}_3^{\afa \beta}\right), \quad  \afa, \beta= 1,2,3,4
	\end{equation}
	satisfies $\mathcal{P}_3^{\beta\alpha} = -\left(\mathcal{P}_3^{\alpha\beta}\right)^\dagger$
	with components as follows:
	\begin{align*}
		\mathcal{P}_3^{11}&=u_{1}\Big(\varLambda u_{2}-u_{2}\varLambda^{-1}+(w\varLambda^{-1}-\varLambda w)u_{1}\Big)+(\varLambda u_{2}-u_{2}\varLambda^{-1})u_{1},\\
		\mathcal{P}_3^{12}&=u_{1}\Big(u_{1}\varDelta u_{2}+(\varLambda^{2}u_{3}-u_{3}\varLambda^{-1})+(w\varLambda^{-1}-\varLambda w)\varGamma u_{2}\Big)+(\varLambda u_{2}-u_{2}\varLambda^{-1})\varGamma u_{2},\\
		\mathcal{P}_3^{13}&=u_{1}\Big(u_{1}\varDelta\varGamma u_{3}+(w\varLambda^{-1}-\varLambda w)\varPhi u_{3}\Big)+(\varLambda u_{2}-u_{2}\varLambda^{-1})\varPhi u_{3},\\
		\mathcal{P}_3^{14}&=u_{1}\Big(u_{1}\varDelta w+(w\varLambda^{-1}-\varLambda w)\varGamma w\Big)+(\varLambda u_{2}-u_{2}\varLambda^{-1})\varGamma w+(u_{3}\varLambda^{-1}-\varLambda^{2}u_{3}),\\
		\mathcal{P}_3^{22}&=u_{2}\varOmega\Big(u_{1}\varDelta u_{2}+(\varLambda^{2}u_{3}-u_{3}\varLambda^{-1})+(w\varLambda^{-1}-\varLambda w)\varGamma u_{2}\Big)+\Big(\varLambda u_{3}-u_{3}\varLambda^{-2}-u_{2}\varTheta u_{1}\Big)\varGamma u_{2},\\
		\mathcal{P}_3^{23}&=u_{2}\varOmega\Big(u_{1}\varDelta\varGamma u_{3}+(w\varLambda^{-1}-\varLambda w)\varPhi u_{3}\Big)+\Big(\varLambda u_{3}-u_{3}\varLambda^{-2}-u_{2}\varTheta u_{1}\Big)\varPhi u_{3},\\
		\mathcal{P}_3^{24}&=u_{2}\varOmega\Big(u_{1}\varDelta w+(w\varLambda^{-1}-\varLambda w)\varGamma w\Big)+\Big(\varLambda u_{3}-u_{3}\varLambda^{-2}-u_{2}\varTheta u_{1}\Big)\varGamma w+(u_{3}\varLambda^{-1} u_{1}-u_{1}\varLambda u_{3}),\\
		\mathcal{P}_3^{33}&=u_{3}\varPsi\Big(u_{1}\varDelta\varGamma u_{3}+(w\varLambda^{-1}-\varLambda w)\varPhi u_{3}\Big)+u_{3}\varTheta\varOmega u_{1}\varPhi u_{3},\\
		\mathcal{P}_3^{34}&=u_{3}\varPsi\Big(u_{1}\varDelta w+(w\varLambda^{-1}-\varLambda w)\varGamma w\Big)+u_{3}\varTheta\varOmega u_{1}\varGamma w+u_{3}\varTheta u_{2},\\
		\mathcal{P}_3^{44}&=w\varOmega\Big(u_{1}\varDelta w+(w\varLambda^{-1}-\varLambda w)\varGamma w\Big)-w\varTheta u_{1}\varGamma w+(u_{3}\varLambda^{-1}-\varLambda u_{3}).
	\end{align*}
\end{thm}
\begin{proof}
	Using Lemmas \ref{var-derivative} and \ref{abk}, we have
	\begin{align*}
		\veps \pfrac{u_1}{t^{4,m+2}}
		&=\frac{1}{(m+3)!}\bigg((\varLambda u_{2}-u_{2}\varLambda^{-1})k^{m+2}_{m+2}+(\varLambda^{2}u_{3}-u_{3}\varLambda^{-1})(k^{m+2}_{m+1})^{-}+u_{1}\varDelta k^{m+2}_{m+3}\bigg),\nonumber\\
		\veps \pfrac{u_2}{t^{4,m+2}}
		&=\frac{1}{(m+3)!}\bigg((-u_{2}\varTheta u_{1}+\varLambda u_{3}-u_{3}\varLambda^{-2})k^{m+2}_{m+2}+(u_{1}\varLambda u_{3}-u_{3}\varLambda^{-1}u_{1})(k^{m+2}_{m+1})^{-}\nonumber\\&\qquad \qquad \qquad+u_{2}\varOmega \varDelta  k^{m+2}_{m+3}\bigg),\nonumber\\
		\veps \pfrac{u_3}{t^{4,m+2}}
		&=\frac{1}{(m+3)!}\bigg(-u_{3}\varTheta\varOmega u_{1}k^{m+2}_{m+2}-u_{3}\varTheta u_{2}(k^{m+2}_{m+1})^{-}+u_{3} \varPsi \varDelta k^{m+2}_{m+3}\bigg),\nonumber\\
		\veps \pfrac{w}{t^{4,m+2}}
		&=\frac{1}{(m+3)!}\bigg(-w\varTheta u_{1}k^{m+2}_{m+2}-w\varDelta u_{3}(k^{m+2}_{m})^{--}+w \varOmega \varDelta k^{m+2}_{m+3}\bigg).
	\end{align*}
	Once $(k^{m+2}_{m})^{--}$, $k^{m+2}_{m+2}$, and $k^{m+2}_{m+3}$ are expressed in terms of $k^{m+1}_{m+1}$, $(k^{m+1}_{m})^{-}$, $(k^{m+1}_{m-1})^{--}$, and $(k^{m+2}_{m+1})^{-}$, the proof of Theorem~\ref{thm:HamFormalism-3} follows by straightforward computations.
	
	Considering Lemma~\ref{abk} for the cases \(r=m+1\), \(r=m+2\), and \(r=m+3\), we obtain
	\begin{align*}
			\begin{split}
			k^{m+2}_{m+2}
			&= u_1 k^{m+1}_{m+1} + \varGamma u_2(k^{m+1}_{m})^{-} + \varPhi u_3(k^{m+1}_{m-1})^{--}- \varGamma w(k^{m+2}_{m+1})^{-},
		\end{split}\\[6pt]
		\begin{split}
		\varDelta k^{m+2}_{m+3}
		&= \bigl( \varLambda u_2 - u_2 \varLambda^{-1} - \varLambda u_1 w + w\varLambda^{-1}u_1 \bigr) k^{m+1}_{m+1} \\
		&\quad + \bigl( u_1\varDelta u_2 + \varLambda^2 u_3 - u_3\varLambda^{-1} - (\varLambda w - w\varLambda^{-1})\varGamma u_2 \bigr) (k^{m+1}_{m})^{-} \\
		&\quad + \bigl( u_1\varDelta\varGamma u_3 - (\varLambda w - w\varLambda^{-1})\varPhi u_3 \bigr) (k^{m+1}_{m-1})^{--} \\
		&\quad - \bigl( u_1\varDelta w - (\varLambda w - w\varLambda^{-1})\varGamma w \bigr) (k^{m+2}_{m+1})^{-},
	\end{split}  \\[6pt]
		\begin{split}
		-w\varDelta u_3(k^{m+2}_m)^{--}
		&= (u_3\varLambda^{-2} - \varLambda u_3) k^{m+1}_{m+1} + (u_3\varLambda^{-1}u_1 - u_1 \varLambda u_3) (k^{m+1}_{m})^{-} \\
		&\quad - u_2\varDelta u_{3}(k^{m+1}_{m-1})^{--} - (u_3\varLambda^{-1} - \varLambda u_3) (k^{m+2}_{m+1})^{-}.
	\end{split}
	\end{align*}
	Substituting the above identities into the $\pp{t^{4,m+2}}$-flows, we complete the proof of Theorem \ref{thm:HamFormalism-3}.
\end{proof}

  		\section{A tri-Hamiltonian structure of the $(3,1)$-type gAL hierarchy}
  	In this section, we show that the operators $(\mathcal{P}_1, \mathcal{P}_2, \mathcal{P}_3)$ given in \eqref{HamP1}--\eqref{HamP3} form a tri-Hamiltonian structure, and compute all central invariants of the resulting bi-Hamiltonian structures.
  			\subsection{Brief review of supervariable formalism}
We first review the definition of the Schouten bracket on the space of local functionals over the supermanifold $\hat{M} = \Pi(T^*M)$, the fiber's parity-reversed cotangent bundle of a smooth $n$-dimensional manifold $M$. More details can be found in \cite{Liu lecture notes,Bihamiltonian cohomologies,Jacobi structure,Kersten}.

Let $\hat{U} = U \times \mathbb{R}^{0|n}$ be a local trivialization of $\hat{M}$, with local coordinates $u^1, \dots, u^n$ on $U$ and dual fermionic coordinates (supervariables) $\theta_1, \dots, \theta_n$ on $\mathbb{R}^{0|n}$.
The infinite jet space $J^\infty(\hat{M})$ admits a local coordinate system
\begin{equation*}
	\left\{ u^{\alpha, s} = \partial_x^s u^\alpha, \theta_\alpha^s = \partial_x^s \theta_\alpha \mid \alpha = 1, \dots, n,\ s \geq 0 \right\},
\end{equation*}
 where all supervariables satisfy the anti-commutation relations
\begin{equation*}
	\theta_\alpha^s \theta_\beta^t = - \theta_\beta^t \theta_\alpha^s, \quad \forall \alpha,\beta = 1,\dots,n,\ s,t \geq 0.
\end{equation*}

Define the ring of differential polynomials
\begin{equation*}
	\hat{\mathcal{A}} = C^\infty(U)\left[\left[ u^{\alpha,s+1}, \theta_{\alpha}^s \mid \alpha = 1,\dots,n,\ s \geq 0 \right]\right],
\end{equation*}
 where the differential operator 
\begin{equation*}
\partial_x = \sum_{s\geq 0} \left( u^{\alpha,s+1} \frac{\partial}{\partial u^{\alpha,s}} + \theta_{\alpha}^{s+1} \frac{\partial}{\partial \theta_{\alpha}^s} \right).
\end{equation*}
 Define the space of local functionals on $\hat{M}$ as the quotient space
 \begin{equation*}
 	\hat{\mathcal{F}} = \hat{\mathcal{A}} / \partial_x \hat{\mathcal{A}},
 \end{equation*}
 where two differential polynomials are considered equivalent if their difference is a total $x$-derivative. We formally denote each equivalence class by $H=\int h \, \mathrm{d}x$, where $h \in \hat{\mathcal{A}}$ is called a density of the local functional $H$.
We equip $\hat{\mathcal{A}}$ with a $\mathbb{Z}$-grading called the super degree, defined by
$$
\deg \theta_{\alpha}^s = 1, \quad \deg u^{\alpha,s} = \deg f = 0, \quad \forall f \in C^\infty(U).
$$
This grading descends naturally to $\hat{\mathcal{F}}$, and a local functional $H \in \hat{\mathcal{F}}$ is homogeneous of super degree $p$ if it admits a density of that degree. For any homogeneous local functional $H=\int h \, \mathrm{d}x \in \hat{\mathcal{F}}$, its variational derivatives with respect to $u^\alpha$ and $\theta_\alpha$ are defined respectively by
$$
\frac{\delta H}{\delta u^\alpha} = \sum_{s\geq 0} (-\partial_x)^s \frac{\partial h}{\partial u^{\alpha,s}}, \quad \frac{\delta H}{\delta \theta_{\alpha}} = \sum_{s\geq 0} (-\partial_x)^s \frac{\partial h}{\partial \theta_{\alpha}^s}.
$$

 On the space of local functionals $\hat{\mathcal{F}}$, the Schouten bracket is defined as follows. For two homogeneous local functionals $G \in \hat{\mathcal{F}}^p$ and $H \in \hat{\mathcal{F}}^q$, their Schouten bracket $[G,H] \in \hat{\mathcal{F}}^{p+q-1}$ is given by
 \begin{equation}
 	[G,H] = \int \left( \frac{\delta G}{\delta \theta_{\alpha}} \frac{\delta H}{\delta u^{\alpha}} + (-1)^{p} \frac{\delta G}{\delta u^{\alpha}} \frac{\delta H}{\delta \theta_{\alpha}} \right) \mathrm{d}x,
 	\label{eq:sn_bracket_discrete}
 \end{equation}
 with the Einstein summation convention in effect for indices $\alpha$.
 
 To establish the correspondence between the Schouten bracket and Hamiltonian operators, we consider a skew-symmetric matrix-valued differential operator $\mathcal{P} = (\mathcal{P}^{\alpha\beta})$, with each component
 \begin{equation*}
 	\mathcal{P}^{\alpha\beta} = \sum_{s \geq 0} \mathcal{P}_s^{\alpha\beta} \partial_x^s.
 \end{equation*}
 The skew-symmetry condition reads $\mathcal{P}^\dagger = -\mathcal{P}$, where $\mathcal{P}^\dagger$ denotes the formal adjoint of $\mathcal{P}$.
 Each such operator $\mathcal{P}$ corresponds to a local functional
 \begin{equation*}
 	\iota(\mathcal{P}) = \frac{1}{2} \int \theta_\alpha (\mathcal{P}^{\alpha\beta} \theta_\beta) \, \mathrm{d}x.
 \end{equation*}
A fundamental result in the supervariable formalism states that a skew-symmetric differential operator $\mathcal{P}$ is a Hamiltonian operator if and only if it satisfies $[\iota(\mathcal{P}), \iota(\mathcal{P})] = 0$.

Based on this, compatible Hamiltonian structures admit
\begin{itemize}
	\item A pair of Hamiltonian operators $(\mathcal{P}_1, \mathcal{P}_2)$ forms a bi-Hamiltonian structure if they are compatible, i.e., 
\begin{equation*}
	[\iota(\mathcal{P}_1), \iota(\mathcal{P}_2)] = 0;
	\end{equation*}
	\item A triple of Hamiltonian operators $(\mathcal{P}_1, \mathcal{P}_2, \mathcal{P}_3)$ forms a tri-Hamiltonian structure if they are pairwise compatible, i.e.,
	\begin{equation*}
		[\iota(\mathcal{P}_1), \iota(\mathcal{P}_2)] = [\iota(\mathcal{P}_1), \iota(\mathcal{P}_3)] = [\iota(\mathcal{P}_2), \iota(\mathcal{P}_3)] = 0.
	\end{equation*}
\end{itemize}

  \subsection{A rigorous proof of the tri-Hamiltonian structure} 
  This subsection is devoted to utilizing the supervariable technique to prove that $\mathcal{P}_1$, $\mathcal{P}_2$ and $\mathcal{P}_3$ given in \eqref{HamP1}--\eqref{HamP3} are Hamiltonian structures, and further verify their compatibility. We first introduce the following notations.
  
  Let $(\theta_1, \theta_2, \theta_3, \theta_4)$ be the dual supervariables associated with $(u_1, u_2, u_3, w)$. For brevity, we write $\theta_\alpha^\pm := \varLambda^\pm \theta_\alpha$ for $\alpha = 1, 2, 3, 4$, and set $F_a = \iota(\mathcal{P}_a)$ ($a=1,2,3$) as the local functionals associated with $\mathcal{P}_1,\mathcal{P}_2$ and $\mathcal{P}_3$.
  \begin{thm}\label{HamP1 structure}
  	$\mathcal P_{1}$ is a Hamiltonian structure, i.e., $$[F_1, F_1] = 0.$$
  \end{thm}
  \begin{proof}
  	By the property that $\int \varDelta u \, \mathrm{d}x = 0$ holds for all $u \in \hat{\mathcal{A}}$, the local functional $F_1$ associated with $\mathcal{P}_1$ can be written as
  	\begin{align*}
  		F_1 &= \frac{1}{2} \int (\theta_1, \theta_2, \theta_3, \theta_4) 
  		\mathcal P_{1}
  		(\theta_1, \theta_2, \theta_3, \theta_4)^\mathrm{T}  \, \mathrm{d}x\\
  		&= \int \Bigl(\theta_1 \varDelta\textbf{A}+w\theta_2\varDelta u_3\theta_3+w^{+}\theta_1 \theta_1^{+}+u_3^{+}\theta_2 \theta_2^{+} \Bigr) \, \mathrm{d}x,
  	\end{align*}
  	where $\textbf{A}=u_2\theta_2 + w\theta_4 + \varGamma u_3\theta_3 $ . 
  	
  	By a straightforward calculation, we obtain the variational derivatives of $F_1$ as
  	\begin{align*}
  		&\frac{\delta F_1}{\delta u_1}=0,\qquad \frac{\delta F_1}{\delta u_2}=-\theta_{1}\theta_{2}+\theta_{1}^{-}\theta_{2},\\
  		&\frac{\delta F_1}{\delta u_3}=-\theta_{1}\theta_{3}-w\theta_{2}\theta_{3}+\theta_{2}^{-}\theta_{2}+w^{-}\theta_{2}^{-}\theta_{3}+\theta_{1}^{--}\theta_{3},\\
  		&\frac{\delta F_1}{\delta w}=-\theta_{1}\theta_{4}-u_3\theta_{2}\theta_{3}+u_3^{+}\theta_{2}\theta_{3}^{+}+\theta_{1}^{-}\theta_{1}+\theta_{1}^{-}\theta_{4},\\
  		&\frac{\delta F_1}{\delta \theta_{1}}=w^{+}\theta_{1}^{+}-w\theta_{1}^{-}+\varDelta\textbf{A},\\
  		&\frac{\delta F_1}{\delta \theta_{2}}=u_2\theta_{1}+u_{3}^{+}\theta_{2}^{+}-u_{3}w\theta_{3}+u_{3}^{+}w\theta_{3}^{+}-u_{2}\theta_{1}^{-}-u_{3}\theta_{2}^{-},\\
  		&\frac{\delta F_1}{\delta \theta_{3}}=u_{3}\theta_{1}+u_{3}w\theta_{2}-u_{3}w^{-}\theta_{2}^{-}-u_{3}\theta_{1}^{--},\qquad \frac{\delta F_1}{\delta \theta_{4}}=w\theta_{1}-w\theta_{1}^{-}.
  	\end{align*}    	
  	By the identities $\int \varDelta u \td x=0$ for all $u\in\hat\mcalA$ and
  	$\theta \theta=0$ for all $\theta\in\hat{\mcalA}^1$,
  	a straightforward computation gives
  	\begin{align*}
  		[F_1,F_1]&=2\int \Biggl(\frac{\delta F_1}{\delta \theta_{1}}\frac{\delta F_1}{\delta u_{1}}+\frac{\delta F_1}{\delta \theta_{2}}\frac{\delta F_1}{\delta u_{2}}+\frac{\delta F_1}{\delta \theta_{3}}\frac{\delta F_1}{\delta u_{3}}+\frac{\delta F_1}{\delta \theta_{4}}\frac{\delta F_1}{\delta w}\Biggr)\mathrm{d}x\\&=2\int \Big(\varTheta(u_{3}^{+}\theta_{1}^{-}\theta_{2}\theta_{2}^{+}-u_3^{+}\theta_{1}\theta_{2}\theta_{2}^{+})\Big)\mathrm{d}x=0.
  	\end{align*}
  	This completes the proof of Theorem \ref{HamP1 structure}.
  \end{proof}
  \begin{thm}\label{HamP12 structure}
  	$\mathcal P_{2}$ is a Hamiltonian structure compatible with $\mathcal P_{1}$, i.e., $$[F_2, F_2] = [F_1, F_2] = 0.$$
  \end{thm}
  \begin{proof}
  	The local functional $F_2$ associated with $\mathcal{P}_2$ admits the form   	
  	\begin{align*}
  		F_2 =& \frac{1}{2} \int \left( \theta_1, \theta_2, \theta_3, \theta_4 \right) 
  		\mathcal{P}_2
  		\left( \theta_1, \theta_2, \theta_3, \theta_4 \right)^\mathrm{T}  \, \mathrm{d}x \\
  		=& \int \bigl(u_2^{+}\theta_1 \theta_1^{+}- u_3^{+}\theta_1^{+}\theta_2
  		+ u_3^{++}\theta_1\theta_2^{++} 
  		+ u_1u_3^{+}\theta_2\theta_2^{+}
  		- u_2 u_3\theta_2 \theta_3+u_3^+ \theta_3^{+}w\theta_4\\&\qquad+\textbf{B}  \textbf{A}^{-}+u_1 \theta_{1}  \varTheta\textbf{A}
  		\bigr) \mathrm{d}x,
  	\end{align*}
  	where $\textbf{B}=u_2\theta_2 + u_3\theta_3+ w\theta_4 $.
  	
  	Using the same method, the variational derivatives of $F_2$ are given by
  	\begin{align*}
  		\frac{\delta F_2}{\delta u_1} &= u_3^+\theta_2\theta_2^+ + \theta_1  \varDelta\textbf{A}, 
  		\qquad
  		\frac{\delta F_2}{\delta u_2} = \theta_1^-\theta_1 +\theta_2\varDelta u_3 \theta_3- \theta_2 \varTheta \textbf{C}, \\
  		\frac{\delta F_2}{\delta u_3} &= \theta_1^{--}\theta_2 - \theta_1\theta_2^- + u_1^-\theta_2^-\theta_2 + \theta_3  \textbf{D}, 
  		\qquad
  		\frac{\delta F_2}{\delta w} = -\theta_4  \varTheta \textbf{C}, \\
  		\frac{\delta F_2}{\delta \theta_1} &= u_2^+\theta_1^+ + u_3^{++}\theta_2^{++} - u_2\theta_1^- - u_3\theta_2^- + u_1  \varDelta\textbf{A}, 
  		\\
  		\frac{\delta F_2}{\delta \theta_2} &= u_3^+\theta_1^+ - u_3\theta_1^{--} + u_1u_3^+\theta_2^+ - u_1^-u_3\theta_2^- +u_2\varDelta u_3 \theta_3- u_2 \varTheta \textbf{C}, 
  		\\
  		\frac{\delta F_2}{\delta \theta_3}& = u_3  \textbf{D}, \qquad
  		\frac{\delta F_2}{\delta \theta_4} =- w  \varTheta\textbf{C},
  	\end{align*}
  	where
  	\begin{align*}
  		\textbf{C}=u_1\theta_1 + \varGamma\bigl(u_2\theta_2 + w\theta_4\bigr) + \varPhi u_3\theta_3,\quad
  		\textbf{D}=\varLambda^{-2} \varDelta \big(\varGamma u_1 \theta_1+\varGamma^2 u_2 \theta_{2}+\varGamma \varPhi u_3 \theta_3+\varPhi w\theta_4 \big).
  	\end{align*}       
  	A direct computation gives
 	\begin{align*}
	[F_2,F_2]
	&= 2\int \Biggl(
	\frac{\delta F_2}{\delta \theta_{1}}\frac{\delta F_2}{\delta u_{1}}
	+\frac{\delta F_2}{\delta \theta_{2}}\frac{\delta F_2}{\delta u_{2}}
	+\frac{\delta F_2}{\delta \theta_{3}}\frac{\delta F_2}{\delta u_{3}}
	+\frac{\delta F_2}{\delta \theta_{4}}\frac{\delta F_2}{\delta w}
	\Biggr) \mathrm{d}x \\
	&= \int \Bigl(
	\varTheta\bigl(u_2^{+}\theta_{1} \theta_{1}^{+}\varDelta(\textbf{B}+\varLambda u_3 \theta_3)+u_3^{+}\theta_{1}^{+}\theta_2(\varLambda \varTheta \varOmega \textbf{B}-\varDelta \varGamma u_3\theta_3-\varTheta u_1 \theta_{1}) \\
	&\qquad \quad+u_1u_3^{+}\theta_2\theta_2^{+}(u_3^{+}\theta_3^{+}+w\theta_4-\varTheta u_1 \theta_1-\varLambda^{-1}\textbf{B})+u_2u_3^{+}\theta_1^{-}\theta_{2}\theta_2^{+}\\[6pt]
	&\qquad \quad-u_3^{+}u_3^{++}\theta_2 \theta_2^{+}\theta_2^{++} -u_3^{+}\theta_1^{-}\theta_1\theta_1^{+}-u_1u_3^{+}\theta_1^{-}\theta_1\theta_2^{+}\bigr)\\[4pt]
	&\qquad \quad+ \varTheta \varOmega \bigl(u_3^{++}\theta_1\theta_2^{++}\varDelta \textbf{A}+u_2u_3u_3^{++}\theta_2 \theta_3 \theta_3^{++}
	\bigr)
	\Bigr) \mathrm{d}x = 0.
\end{align*}
  	Moreover, we have
  	 		\begin{align*}
  		[F_1,F_2]&=\int \Bigg(\sum_{i=1}^{3}\Big(	\frac{\delta F_1}{\delta \theta_{i}}\frac{\delta F_2}{\delta u_i}+
  		\frac{\delta F_1}{\delta u_i}\frac{\delta F_2}{\delta \theta_{i}}\Big)
  		+
  		\frac{\delta F_1}{\delta \theta_{4}}\frac{\delta F_2}{\delta w}	+
  		\frac{\delta F_1}{\delta w}\frac{\delta F_2}{\delta \theta_{4}}
  		\Bigg)\mathrm{d}x\\
  		&=\int\varTheta\big(u_{3}^{+}w\theta_{1}^{-}\theta_{2}\theta_{2}^{+}-u_{3}^{+}w^{+}\theta_{1}^{+}\theta_{2}\theta_{2}^{+}-u_{3}^{+}\theta_{1}\theta_{1}^{+}\theta_{2}-u_{3}^{+}\theta_{1}^{-}\theta_{1}\theta_{2}^{+}\\&\qquad \qquad-u_{2}^{-}u_{3}^{+}\theta_{2}^{-}\theta_{2}\theta_{2}^{+}+u_{3}^{+}\theta_{1}^{-}\theta_{1}^{+}\theta_{2}-u_1u_{3}^{+}\theta_{1}^{-}\theta_{2}\theta_{2}^{+}+u_1u_{3}^{+}\theta_{1}\theta_{2}\theta_{2}^{+}\\[6pt]&\qquad \qquad
  		+w^{+}\theta_{1}\theta_{1}^{+}(\varDelta u_2 \theta_{2}+\varDelta \varGamma u_3\theta_3+\varGamma w\theta_4)-u_3^{+}\theta_2 \theta_2^{+}(\varTheta u_1\theta_1\\[6pt]&\qquad \qquad+\varLambda \varTheta \varOmega u_3 \theta_3+\varTheta w\theta_4)
  		\big)\mathrm{d}x=0.
  	\end{align*}
  	Theorem \ref{HamP12 structure} is hereby proved.    	
  \end{proof}
  \begin{thm}\label{HamP123 structure}
  	$\mathcal P_{3}$ is a Hamiltonian structure compatible with $\mathcal P_{1}$ and $\mathcal P_{2}$, i.e., $$[F_3, F_3] = [F_1, F_3] =[F_2, F_3]= 0.$$
  \end{thm}
  \begin{proof}
  Using the property that $\int \varDelta u \, \mathrm{d}x = 0$ holds for all $u \in \hat{\mathcal{A}}$, the local functional $F_3$ associated with $\mathcal{P}_3$ reads
  	\begin{align*}
  		F_3 &= \frac{1}{2} \int (\theta_1, \theta_2, \theta_3, \theta_4) 
  		\mathcal{P}_{3}
  		(\theta_1, \theta_2, \theta_3, \theta_4)^\mathrm{T} \, \mathrm{d}x \\[6pt]
  		&= \int \biggl(\Bigl( (w + w^- - 2u_1^-)  \mathbf{B} 
  		+ w u_1 \theta_1 
  		- w^- u_3 \theta_3 
  		- u_2 \theta_1 + u_1^- u_3 \theta_3 \Bigr) \mathbf{B}^- \\
  		&\qquad \quad
  		+\Bigl( w(-\varTheta) u_1 \theta_1 
  		- w^- \mathbf{B}^{--} 
  		+ \varGamma\bigl( u_2^- \theta_1^{--} + u_3^- \theta_2^{--} \bigr) 
  		- w^- u_1^{--} \theta_1^{--} 
  		\\[4pt]
  		&\qquad\quad- u_2 \theta_1  - u_1^2 \theta_1 + u_1^- u_3 \theta_3 \Bigr)  \mathbf{B} 
  		- 2u_1^{--} u_3 \theta_3 \mathbf{B}^{--} 
  		+ w^{--} u_3 \theta_3 \mathbf{C}^{---} 
  		- u_3 \theta_2 \mathbf{C}^{--} 
  		\\[4pt]
  		&\qquad \quad- w u_1^- \theta_1^- u_1 \theta_1 
  		+ w^- w \theta_4 u_3 \theta_3  + w^- u_2 \theta_2 u_3 \theta_3 
  		- w^- u_1^{-} \theta_1^{-} u_3 \theta_3  
  		+ u_1^- \theta_1^- u_2 \theta_1 
  		\\[6pt]
  		&\qquad \quad+ u_1^{--} u_3^- \theta_3^- u_3 \theta_3 
  		+ u_3 \theta_2^{-} u_1 \theta_1 
  		+ u_3 \theta_3 \bigl( u_2^{-} \theta_1^- - u_2^{--} \theta_1^{---} - u_3^{--} \theta_2^{---} +u_2^{-}\theta_4^{-}\\[6pt]
  		&\qquad \quad-u_2 \theta_4 \bigr) + u_3 \theta_4 \theta_4^- 
  		- u_3 \theta_1^{--} \theta_4 
  		+ u_1^- u_3 \theta_2 \theta_4^- 
  		- u_1^- u_3 \theta_2^- \theta_4 
  		- u_3 \theta_4^- \theta_4+ u_2 \theta_1^- u_1 \theta_1 \\
  		&\qquad \quad + u_1^- u_1^- \theta_1^- u_2 \theta_2  
  		+ u_1^{--} u_1^{--} \theta_1^{--} u_3 \theta_3 
  		+ u_1^{--} u_1^{--} \theta_1^{--} w \theta_4 
  		\biggr) \, \mathrm{d}x.
  	\end{align*}
  	Then the variational derivatives of $F_3$ are given by 
  	\begin{align*}
  		\frac{\partial F_3}{\partial \theta_1}
  		&= u_1^2 \varDelta \textbf{A}+ (u_1 w - u_2) \textbf{C}^- + (u_2^+ - u_1 w^+)\textbf{C}^+ 
  		- u_1 u_2 \theta_1^- - u_1 u_3 \theta_2^- + u_3 \theta_4^- \\
  		&\quad+ u_1 u_2^+ \theta_1^+ + u_1 u_3^{++} \theta_2^{++} + u_1^2 \varLambda^2 \varTheta w \theta_4, \\[6pt]
  		\frac{\partial F_3}{\partial u_1}
  		&=(2 u_1 \theta_1 + u_1^+ \theta_3^+) \varDelta \textbf{A} - 2 \left( u_2^+ \theta_2^+ + u_3^{++} \theta_3^{++} \right) \textbf{B} +w \theta_1 \textbf{C}^- - w^+ \theta_1 \textbf{C}^+ 
  		+ u_2 \theta_1^- \theta_1 \\
  		&\quad- u_3 \theta_1 \theta_2^- + u_2^+ \theta_1 \theta_1^+ + u_3^{+} \theta_2^+ \theta_4 - u_3^{+} \theta_2 \theta_4^+ + u_3^{++} \theta_1 \theta_2^{++}- 2 u_1 \theta_1 \varLambda^2 \varTheta w \theta_4, \\[6pt]
  		\frac{\partial F_3}{\partial \theta_2}
  		&= (u_2 w^- - u_3) \textbf{C}^{--}- u_2 w (-\varTheta) \textbf{C} + (u_3^+ - u_2 w^+) \textbf{C}^+ 
  		+ u_1 u_2 \textbf{E} - u_1^- u_2 \textbf{F} \\
  		&\quad + u_2 \varGamma \textbf{I} - u_3^{-} u_3 \theta_2 \theta_3^- + u_1^- u_3 \theta_4^- - u_1 u_3^+ \theta_4^+ + u_3^- u_3 \theta_3^-, \\[6pt]
  		\frac{\partial F_3}{\partial u_2}
  		&= w^- \theta_2 \textbf{C}^{--}- \theta_1 \textbf{C}^- +   \left(w \theta_2 \varTheta 
  		+ \theta_1^- \right)\textbf{C}  - w^+ \theta_2 \textbf{C}^+ + u_1 \theta_2 \textbf{E} - u_1^- \theta_2 \textbf{F}\\
  		&\quad + \theta_2 \varGamma \textbf{I} - \theta_4 \varDelta u_3 \theta_3, \\[6pt]
  		\frac{\partial F_3}{\partial \theta_3}
  		&=u_1^- u_3 (-\varTheta) \textbf{A}+ u_3 w^{--} \textbf{C}^{---}  + u_3 w \varTheta \textbf{C}
  		- u_3 w^+ \textbf{C}^+ + u_1 u_3 \textbf{E} \\
  		&\quad- u_1^{--} u_3 \textbf{F}^- + u_3 w^- \textbf{G} + u_3 \textbf{H}, \\[6pt]
  		\frac{\partial F_3}{\partial u_3}
  		&= w^{--} \theta_3 \textbf{C}^{---} - \theta_2 \textbf{C}^{--} +\left( w \theta_3 \varTheta 
  		+ \theta_2^- \right)\textbf{C} - w^+ \theta_3 \textbf{C}^+ + u_1 \theta_3 \textbf{E} - u_1^{--} \theta_3 \textbf{F}^- \\
  		&\quad+ w^- \theta_3 \textbf{G} + \theta_3 \textbf{H}
  		+ \theta_1 \theta_4^- - \theta_1^{--} \theta_4 + u_1^- \theta_2 \theta_4^- - u_1^- \theta_2^- \theta_4 - \theta_4^- \theta_4 - u_3^{+} \theta_2^- \theta_3^+,\\[6pt]
  		\frac{\partial F_3}{\partial \theta_4}
  		&= w^- w \textbf{C}^{--}  + w^2 \varTheta \textbf{C}- w w^+ \textbf{C}^+
  		+ u_1 w \textbf{E} - u_1^- w \textbf{F} + w \varGamma \textbf{I} + u_3 \theta_1^{--}
  		- u_1^{--} u_1^{--} w \theta_1^{--}
  		\\
  		&\quad+ u_1^- u_3 \theta_2^- + u_3 \theta_4^- - u_3^+ \theta_1^+
  		- u_1 u_3^+ \theta_2^+ - u_2 u_3^+ \theta_3^+ - u_3^+ \theta_4^+ + u_1^- u_1^- w \theta_1^-, \\[6pt]
  		\frac{\partial F_3}{\partial w}
  		&= \left( u_3^- \theta_3^- + u_1^- \theta_1^- \right) \left(w^{+} \theta_4^{+}-\varGamma \textbf{B} \right)+ w^- \theta_4 \textbf{C}^{--} 
  		+ ( w \theta_4 + w^+ \theta_4^+ + u_3^{++} \theta_3^{++} \\
  		&\quad+ u_1 \theta_1 ) \textbf{C}^-
  		- \left( w \theta_4 + w^- \theta_4^- \right) \textbf{C}- w^+ \theta_4 \textbf{C}^+
  		+ u_1 \theta_4 \textbf{E} - u_1^- \theta_4 \textbf{F}
  		- u_2^- \theta_2^- u_2^+ \theta_2^+
  		\\[3pt]
  		&\quad + w^+ \theta_4 w^- \theta_4^- + w^- \theta_4^- u_3^{++} \theta_3^{++}
  		- u_1 \theta_1 w \theta_4 
  		- \left( u_3 \theta_3 + u_3^+ \theta_3 \right) \varDelta \varOmega u_2 \theta_2\\[4pt]
  		&\quad+ \left( u_2^+ \theta_2^+ - u_2^- \theta_2^- \right) \left( w \theta_4 + u_2 \theta_2 \right)
  		+ (u_1^{--})^2 \theta_1^{--} \theta_4 + u_2^- \theta_1^{--} \theta_4
  		+ u_2 \theta_1^{-} \theta_4 - u_2 \theta_1 \theta_4
  		\\[3pt]
  		&\quad	+ u_3^- \theta_2^{--} \theta_4 + u_3 \theta_2^- \theta_4
  		- u_2^+ \theta_1^+ \theta_4 - u_3^+ \theta_2^+ \theta_4
  		- u_3^{++} \theta_2^{++} \theta_4 + u_1^- \theta_4 u_1^- \theta_1^-,
  	\end{align*}
  	where
  	\begin{align*}
  		\textbf{E} &= u_1 \theta_1 + 2\varLambda \left( u_2 \theta_2 + w \theta_4 \right) + (\varLambda + 2\varLambda^2) u_3 \theta_3, \\[6pt]
  		\textbf{F} &= u_1^- \theta_1^- + 2\varLambda^{-1} \left( u_2 \theta_2 + w \theta_4 \right) + (1 + 2\varLambda^{-1}) u_3 \theta_3, \\[6pt]
  		\textbf{G} &=\varLambda^{-1} \varTheta u_1 \theta_1 + \varTheta \varOmega \left( u_2 \theta_2 + w \theta_4 \right)-\varDelta \varPsi u_3 \theta_3, \\[6pt]
  		\textbf{H} &= \varLambda^{-1}\varPhi u_2 \theta_1 + \varPhi \left( u_3 \theta_2 - u_2^{--} \theta_1^{---} - u_3^{--} \theta_2^{---} \right) + \varTheta u_2 \theta_4, \\[6pt]
  		\textbf{I} &= u_2 \theta_1 - u_2^- \theta_1^{--} + u_3^+ \theta_2^+ - u_3^- \theta_2^{--}.
  	\end{align*}      
  	Through direct calculation and using the identities $\int \varDelta u \td x=0$ for all $u\in\hat\mcalA$ and
  	$\theta \theta=0$ for all $\theta\in\hat{\mcalA}^1$, we obtain
  	\begin{align*}
  		[F_3,F_3]
  		= 2\int \Biggl(
  		\frac{\delta F_3}{\delta \theta_{1}}\frac{\delta F_3}{\delta u_{1}}
  		+\frac{\delta F_3}{\delta \theta_{2}}\frac{\delta F_3}{\delta u_{2}}
  		+\frac{\delta F_3}{\delta \theta_{3}}\frac{\delta F_3}{\delta u_{3}}
  		+\frac{\delta F_3}{\delta \theta_{4}}\frac{\delta F_3}{\delta w}
  		\Biggr) \mathrm{d}x =0.
  	\end{align*}
  	Together with straightforward calculations, we have
  	\begin{align*}
  		[F_1,F_3]=\int \Biggl(\sum_{i=1}^{3}\Big(	\frac{\delta F_1}{\delta \theta_{i}}\frac{\delta F_3}{\delta u_i}+
  		\frac{\delta F_1}{\delta u_i}\frac{\delta F_3}{\delta \theta_{i}}\Big)
  		+
  		\frac{\delta F_1}{\delta \theta_{4}}\frac{\delta F_3}{\delta w}	+
  		\frac{\delta F_1}{\delta w}\frac{\delta F_3}{\delta \theta_{4}}
  		\Biggr)\mathrm{d}x=0
  	\end{align*}
  	and
  	\begin{align*}
  		[F_2,F_3]=\int \Biggl(\sum_{i=1}^{3}\Big(	\frac{\delta F_2}{\delta \theta_{i}}\frac{\delta F_3}{\delta u_i}+
  		\frac{\delta F_2}{\delta u_i}\frac{\delta F_3}{\delta \theta_{i}}\Big)
  		+
  		\frac{\delta F_2}{\delta \theta_{4}}\frac{\delta F_3}{\delta w}	+
  		\frac{\delta F_2}{\delta w}\frac{\delta F_3}{\delta \theta_{4}}
  		\Biggr)\mathrm{d}x=0.
  	\end{align*}
  This completes the proof of Theorem~\ref{HamP123 structure}. 
  \end{proof}
 
 \begin{rmk}
 	The proof for Theorem~\ref{HamP123 structure} is analogous to those of Theorems~\ref{HamP1 structure} and~\ref{HamP12 structure}. However, the three Schouten brackets involved in Theorem~\ref{HamP123 structure} each contain tens of thousands of terms; we therefore omit the computational details. Nevertheless, all calculations for Theorems~\ref{HamP1 structure}--\ref{HamP123 structure} have been carried out using symbolic computation, which confirms the validity of the tri-Hamiltonian structure for the $(3,1)$-type gAL hierarchy.
 \end{rmk}
    		\subsection{Central invariants of the associated bi-Hamiltonian structures}
Hydrodynamic-type bi-Hamiltonian structures are intrinsically connected to 2D topological field theory and Gromov–Witten invariants. To better understand the corresponding bihamiltonian integrable hierarchies, Dubrovin and Zhang initiated a systematic classification of deformations of bi-Hamiltonian structures in their foundational work \cite{Normal forms}. This framework was later completed by Dubrovin, Liu and Zhang, who introduced central invariants for any deformation of a given semisimple hydrodynamic-type bi-Hamiltonian structure \cite{c1,c2}. More precisely, an $n$-dimensional semisimple Frobenius manifold corresponds to a family parameterized by $n$ single-variable functions, which are precisely the associated central invariants.

    We begin by recalling the definition of central invariants. 
    Let $M$ be an $n$-dimensional smooth manifold with local coordinates $\mathbf{u}=(u^1,\dots, u^n)$, and let $(\mathcal{Q}_1,\mathcal{Q}_2)$ be a bi-Hamiltonian structure on the jet space $J^\infty(M)$ with components
    \begin{align*}
    	\mathcal{Q}_a^{ij} &=
    	\sum_{k\geq 0}\varepsilon^k
    	\left(
    	\sum_{m=0}^{k+1}
    	\mathcal{Q}_{k,m;a}^{ij}
    	\partial_x^m
    	\right)
    	\\
    	&=
    	g_a^{ij}(\mathbf{u})\partial_x + \Gamma_{a;\gamma}^{ij}(\mathbf{u}) u_x^\gamma
    	+\varepsilon\left(
    	H_a^{ij}(\mathbf{u})\partial_x^2 + \cdots
    	\right)
    	+\varepsilon^2\left(
    	K_a^{ij}(\mathbf{u})\partial_x^3+\cdots
    	\right)
    	+O(\varepsilon^3),
    \end{align*}
    where $a=1,2$, and $\mathcal{Q}_{k,m; a}^{ij}$ are homogeneous differential polynomials of degree $k+1-m$ on $J^\infty(M)$. The leading coefficients $g_a^{ij}(\mathbf{u})$ define two flat non-degenerate contravariant metrics on $M$, 
    and $\Gamma_{a;\gamma}^{ij}(\mathbf{u})$ are the corresponding contravariant Christoffel symbols 
    of the Levi-Civita connections associated with the metrics $g_a^{ij}(\mathbf{u})$.  The bi-Hamiltonian structure $(\mathcal{Q}_1,\mathcal{Q}_2)$ has a semisimple hydrodynamic limit 
    if the characteristic equation
  $\det(g_2^{ij}- \lambda g_1^{ij})=0$
    possesses $n$ distinct, non-constant eigenvalues $(\lambda^1, \dots, \lambda^n)$ near a generic point of $M$. 
    These roots form a system of local coordinates on $M$, 
    known as the canonical coordinates of the semisimple bi-Hamiltonian structure \cite{Ferapontov}.
    
     An equivalent characterization of these canonical coordinates is given by the dispersionless Lax superpotential $\lambda(p)$, which admits $n$ pairwise distinct critical points $p_1, \dots, p_n$ satisfying $\lambda'(p)\big|_{p=p_k} = 0$ for all $k=1,\dots,n$, and whose corresponding critical values $\lambda^i = \lambda(p_i)$ coincide precisely with the canonical coordinates of the semisimple bi-Hamiltonian structure. In the canonical coordinates, the flat metrics $g_1$ and $g_2$ take the following diagonal form
    \begin{equation}\label{fi}				
    	g_1^{ij}(\lambda) = \delta^{ij}f^i, \quad g_2^{ij}(\lambda)  = \delta^{ij}\lambda^i f^i,
    \end{equation}
    where $\delta^{ij}$ denotes the Kronecker delta symbol, $f^i$ are smooth functions on $M$, and
    	\begin{equation*}
    	g_a^{ij}(\lambda) = \frac{\partial \lambda^i}{\partial u^k} g_a^{kl}(\mathbf{u}) \frac{\partial \lambda^j}{\partial u^l}, \quad a = 1,2.
    \end{equation*}	  
    Let the Jacobian matrix be		$J=(J^i_\afa)=\left(\pfrac{\lmd^i}{u^\afa}
    \right)$, then
    	\[J g_1 J^{\mathrm{T}}=
    \begin{pmatrix}
    	f^1 & 0 & \cdots & 0 \\
    	0 & f^2 & \cdots & 0 \\
    	\vdots & \vdots & \ddots & \vdots \\
    	0 & 0  & \cdots & f^n 
    \end{pmatrix} ,\qquad
  J g_2 J^{\mathrm{T}}=
    \begin{pmatrix}
    	\lambda^1 f^1 & 0 & \cdots & 0 \\
    	0 & \lambda^2 f^2 & \cdots & 0 \\
    	\vdots & \vdots & \ddots & \vdots \\
    	0 & 0  & \cdots & \lambda^n f^n 
    \end{pmatrix}.
    \]
   
   Further, we are particularly concerned with the leading coefficients of the first-order and second-order deformation terms  
    	$ H_a=(H_a^{ij})$,  $K_a=(K_a^{ij})$, 
    		which satisfy the skew-symmetry $H_a^{ij} = -H_a^{ji}$ and symmetry $K_a^{ij} = K_a^{ji}$, respectively. Their coordinate transformations to the canonical coordinates are given by
    		\begin{equation*}
    		H_a^{ij}(\lambda) = \frac{\partial \lambda^i}{\partial u^k} 	H_a^{kl}(\mathbf{u}) \frac{\partial \lambda^j}{\partial u^l},\quad
    			K_a^{ij}(\lambda) = \frac{\partial \lambda^i}{\partial u^k} K_a^{kl}(\mathbf{u}) \frac{\partial \lambda^j}{\partial u^l}, \quad a = 1,2.
    		\end{equation*}
    		Then the central invariants of the semisimple bi-Hamiltonian structure $(\mathcal{Q}_{1},\mathcal{Q}_{2})$ are defined as
    	\begin{equation}\label{central invariants}
    		c_i(\lambda) = \frac{1}{3(f^i)^2}
    		\left( K_2^{ii}(\lambda) - \lambda^iK_1^{ii}(\lambda) +
    		\sum_{k\neq i} \frac{( H_2^{ki}(\lambda) - \lambda^iH_1^{ki}(\lambda))^2}{f^k (\lambda^k-\lambda^i)} \right)
    	\end{equation}		
    	with respect to the canonical coordinates $(\lambda^1,\dots,\lambda^n)$. To simplify the calculation, we follow the method in \cite{c3} and introduce the parameter $\xi$ along with the following notation
    \begin{equation*}
    	g_\xi := g_2 - \xi g_1, \qquad H_\xi := H_2 - \xi H_1, \qquad K_\xi := K_2 - \xi K_1,
    \end{equation*}
    and the tensor
    \begin{equation*}
    	A_\xi := K_\xi + \frac{1}{2} H_\xi^{\mathrm{T}} g_\xi^{-1} H_\xi.
    \end{equation*}
    Then the central invariants \eqref{central invariants} admit the following equivalent representation 
    \begin{equation}
    	c_i (\lambda)= -\frac{1}{3 f^i} \operatorname*{Res}_{\xi=\lambda^i} \bigl( g_\xi^{-1} A_\xi \bigr).
    \end{equation}
    
    An interesting result is that the equivalence classes of infinitesimal deformations of semisimple hydrodynamic-type bi-Hamiltonian structures under a Miura transformation are uniquely determined by the central invariants \cite{unobstructed}. In other words, two deformations of a semisimple bi-Hamiltonian structure $(\mcalP_1,\mcalP_2)$ are equivalent under a Miura transformation if and only if they have the same central invariants. In addition, integrable systems with central invariants equal to $\frac{1}{24}$ are well-behaved: besides having a bi-Hamiltonian structure, it can be proved that they admit a tau function, and satisfy Virasoro symmetry, etc. An important example is the bi-Hamiltonian integrable hierarchy governed by the cohomological field theory associated with a semisimple Frobenius manifold. In this case, all central invariants of its bi-Hamiltonian structure are equal to $\frac{1}{24}$ \cite{c1,c2,Central invariants}. Such a bi-Hamiltonian structure can be viewed as a topological deformation of its dispersionless limit \cite{Normal forms}.
     	
  Next, we calculate all central invariants of the bi-Hamiltonian structures $(\mathcal{P}_i, \mathcal{P}_j)$ ($1 \leq i < j \leq 3$) of the $(3,1)$-type gAL hierarchy; the results are summarized in Theorem \ref{cv-table}. As a concrete example, we provide the detailed proof for the bi-Hamiltonian pair $(\mathcal{P}_1, \mathcal{P}_2)$.
    	
    	\begin{thm}\label{central-invariants thm}
    		The central invariants of the bi-Hamiltonian structure $(\mathcal P_{1},\mathcal P_{2})$ given by \eqref{HamP1}--\eqref{HamP2} for the $(3,1)$-type gAL hierarchy are all equal to $\frac{1}{24}$, i.e.,
    		\begin{align}
    			c_{1}=c_{2}=c_{3}=c_{4}=\frac{1}{24}.
    		\end{align}
    	\end{thm}	
    	\begin{proof}	
    		After the rescaling $\mcalP_a\mapsto\frac{1}{\veps}\mcalP_a$, $a=1,2$,
    		the contravariant metrics $\eta:=g_1$ and
    		$g:=g_2$ corresponding to $(\mcalP_1, \mcalP_2)$
    		have the local form
    			\begin{gather}\label{flat pencil}
    			\renewcommand{\arraystretch}{1.5}
    			\small
    		\eta=
    			\begin{pmatrix}
    				2u^4 & u^2 & 2u^3 & u^4 \\
    				u^2 & 2u^3 & u^3 u^4 & 0 \\
    				2u^3 & u^3 u^4 & 0 & 0 \\
    				u^4 & 0 & 0 & 0
    			\end{pmatrix}, ~
    			g=
    			\begin{pmatrix}
    				2u^2 & 3u^3 + u^1 u^2 & 2u^1 u^3 & u^1 u^4 \\
    				u^1 u^2 + 3u^3 & 2u^1 u^3 + 2(u^2)^2 & 4u^2 u^3 & 2u^2 u^4 \\
    				2u^1 u^3 & 4u^2 u^3 & 6(u^3)^2 & 3u^3 u^4 \\
    				u^1 u^4 & 2u^2 u^4 & 3u^3 u^4 & 2(u^4)^2
    			\end{pmatrix}
    		\end{gather}
    		with respect to the local coordinates $(u^1, u^2, u^3,u^{4}):=(u_{1},u_{2},u_{3},w)$,
    		and by further calculations,
    		the corresponding coefficients $H_1$ and $K_1$ satisfy
    		\begin{gather}
    			\renewcommand{\arraystretch}{1.5}
    			\small
    			H_1 =
    			\begin{pmatrix}
    				0 & \frac{1}{2}u^2 & 2u^3 & \frac{1}{2}u^4 \\
    				-\frac{1}{2}u^2 & 0 & \frac{1}{2}u^3 u^4 & 0 \\
    				-2u^3 & -\frac{1}{2}u^3 u^4 & 0 & 0 \\
    				-\frac{1}{2}u^4 & 0 & 0 & 0
    			\end{pmatrix},
    			\quad
    			K_1 =
    			\begin{pmatrix}
    				\frac{1}{3}u^4 & \frac{1}{6}u^2 & \frac{4}{3}u^3 & \frac{1}{6}u^4 \\
    				\frac{1}{6}u^2 & \frac{1}{3}u^3 & \frac{1}{6}u^3 u^4 & 0 \\
    				\frac{4}{3}u^3 & \frac{1}{6}u^3 u^4 & 0 & 0 \\
    				\frac{1}{6}u^4 & 0 & 0 & 0
    			\end{pmatrix}.
    		\end{gather}
    		Similarly, $H_2$ and $K_2$ can be obtained via a straightforward calculation, whose explicit expressions we omit here.

    		The canonical coordinates $(\lmd^1,\lmd^2,\lmd^3,\lmd^4)$ of $(\mcalP_1, \mcalP_2)$
    		are the four distinct roots of the polynomial 
    $F(X)=\det\left(g_2- X g_1\right)$,
    	and then the Jacobian $J = (J^i_\alpha) = \left( \frac{\partial \lambda^i}{\partial u^\alpha} \right)$ can be computed directly for all $1 \leq i \leq 4$.	 		
    		The diagonal coefficients $f^i$ \eqref{fi} can be computed via
    		\begin{align*}
    			\begin{pmatrix}
    				f^1&0&0&0\\
    				0&f^2&0&0\\
    				0&0&f^3&0\\
    				0&0&0&f^4
    			\end{pmatrix}=
    			J
    				\begin{pmatrix}
    				2u^{4}&u^{2}&2u^{3}&u^{4}\\
    				u^{2}&2u^{3}&u^{3}u^{4}&0\\
    				2u^{3}&u^{3}u^{4}&0&0\\
    				u^{4}&0&0&0
    			\end{pmatrix}
    			J^\mathrm{T}.
    		\end{align*}
    		Since $H_a^{ij}$ and $K_a^{ij}$ are coefficients of $(2,0)$-tensor fields, they can also be transformed to the canonical coordinates via the Jacobian $J=(J^i_\afa)$.
    		According to the definition of central invariants \eqref{central invariants},
    		we calculate all central invariants of the bi-Hamiltonian structure $(\mathcal{P}_{1},\mathcal{P}_{2})$ directly and obtain
    		$$
    		c_{1}=c_{2}=c_{3}=c_{4}=\frac{1}{24}.
    		$$
    		Theorem \ref{central-invariants thm} is proved.
    	\end{proof}
    	In a similar manner, we compute all central invariants of the other bi-Hamiltonian structures $(\mcalP_i, \mcalP_j)$ and obtain the following theorem.
\begin{thm}\label{cv-table}
	The central invariants of the bi-Hamiltonian structures $(\mathcal{P}_i, \mathcal{P}_j)$ of the asymmetric gAL hierarchy of type (3,1) are given by 
	
	\begin{table}[htbp]
		\centering
		\captionsetup{labelfont=bf} 
		\caption{Central invariants for bi-Hamiltonian pairs $(\mathcal{P}_i,\mathcal{P}_j)$.}
		\vspace{2pt}
		\begin{tabular}{lcccc}
			\toprule
			$(\mathcal{P}_i, \mathcal{P}_j)$ & $c_1$ & $c_2$ & $c_3$ & $c_4$ \\
			\midrule
			$(\mathcal{P}_1, \mathcal{P}_2)$ & $\dfrac{1}{24}$ & $\dfrac{1}{24}$ & $\dfrac{1}{24}$ & $\dfrac{1}{24}$    \\[12pt]
			$(\mathcal{P}_2, \mathcal{P}_1)$ & $-\dfrac{1}{24\lambda^1}$ & $-\dfrac{1}{24\lambda^2}$ &
			$-\dfrac{1}{24\lambda^3}$ & $-\dfrac{1}{24\lambda^4}$ \\[12pt]
			$(\mathcal{P}_1, \mathcal{P}_3)$ & $\dfrac{1}{48\sqrt{\lambda^1}}$ & $\dfrac{1}{48\sqrt{\lambda^2}}$ &
			$\dfrac{1}{48\sqrt{\lambda^3}}$ & $\dfrac{1}{48\sqrt{\lambda^4}}$  \\[12pt]
			$(\mathcal{P}_3, \mathcal{P}_1)$ & $\dfrac{1}{48\sqrt{\lambda^1}}$ & $\dfrac{1}{48\sqrt{\lambda^2}}$ &
			$\dfrac{1}{48\sqrt{\lambda^3}}$ & $\dfrac{1}{48\sqrt{\lambda^4}}$ \\[12pt]
			$(\mathcal{P}_2, \mathcal{P}_3)$ & $\dfrac{1}{24\lambda^1}$ & $\dfrac{1}{24\lambda^2}$ &
			$\dfrac{1}{24\lambda^3}$ & $\dfrac{1}{24\lambda^4}$ \\[12pt]
			$(\mathcal{P}_3, \mathcal{P}_2)$ & $-\dfrac{1}{24}$ & $-\dfrac{1}{24}$ & $-\dfrac{1}{24}$ & $-\dfrac{1}{24}$ \\[12pt]
			\bottomrule
		\end{tabular}
	\end{table}
\end{thm}
\newpage
\section{The dispersionless limit of the $(3,1)$-type gAL hierarchy}
In this section, we construct a Frobenius manifold $M$ from the dispersionless limit of the $(3,1)$-type gAL hierarchy, and show that the dispersionless limits of its first flows belong to the Principal Hierarchy of $M$.
\subsection{Frobenius manifold associated with  the $(3,1)$-type gAL hierarchy}   
Frobenius manifolds, introduced by Dubrovin in the early 1990s \cite{Frobenius manifold,Dubrovin1996}, provide a geometric framework for the Witten--Dijkgraaf--Verlinde--Verlinde (WDVV) associativity equations arising in 2D topological field theory \cite{TFT-2,TFT-1}. 
\begin{defn}
Let $M$ be an $n$-dimensional smooth or analytic manifold. We say $M$ is a Frobenius manifold if it admits a Frobenius algebra structure, consisting of a multiplication $\cdot$, a non-degenerate multiplication-invariant bilinear form $\langle\,,\rangle$, and a unity $e$, that depends smoothly or analytically on $p\in M$ and satisfies the following axioms:
	
	\begin{enumerate}	\item The bilinear form $\langle\,,\rangle$ defines a flat metric $\eta$ on $M$, and $\nabla e = 0$, where $\nabla$ is the Levi-Civita connection with respect to $\eta$.	\item Define a $(0,3)$-tensor $c$ on $M$ by $c(X_1,X_2,X_3) := \langle X_1 \cdot X_2, X_3\rangle$. Then the $(0,4)$-tensor $\nabla c$ on $M$, defined by	\[	(\nabla c)(W,X_1,X_2,X_3) = \nabla_W\bigl(c(X_1,X_2,X_3)\bigr),	\]	is required to be symmetric in all four arguments $W,X_1,X_2,X_3 \in \operatorname{Vect}(M)$.	\item There exists an Euler vector field $E$ on $M$ such that $\nabla \nabla E = 0$, $\nabla E$ is diagonalizable, and	\begin{align*}		[E,X_1\cdot X_2] - [E,X_1]\cdot X_2 - X_1\cdot [E,X_2] &= X_1\cdot X_2,\\		E\langle X_1,X_2\rangle - \langle [E,X_1],X_2\rangle - \langle X_1,[E,X_2]\rangle &= (2-d)\langle X_1,X_2 \rangle,	\end{align*}	
where $d$ is a constant called the charge of the Frobenius manifold.\end{enumerate}
\end{defn}

We now recall some basic properties of the Frobenius manifold $M(\eta, c, e, E)$. Fix a system of flat local coordinates $\boldsymbol{z} = (z^1,\dots,z^n)$ for the metric $\eta$, in which $\eta = (\eta^{\alpha\beta})$ takes constant values. We denote
\begin{equation}	
	\eta_{\alpha\beta} = \eta\Bigl(\frac{\partial}{\partial z^\alpha},\frac{\partial}{\partial z^\beta}\Bigr),\qquad	\eta^{\alpha\beta} = (\eta_{\alpha\beta})^{-1}.
	\end{equation}
The 4-symmetry of $\nabla c$ implies the existence of a local function $F$ on $M$, called the Frobenius potential, such that
\begin{equation}
		c\Bigl(\frac{\partial}{\partial z^\alpha},\frac{\partial}{\partial z^\beta},\frac{\partial}{\partial z^\gamma}\Bigr)	= \frac{\partial^3 F}{\partial z^\alpha\partial z^\beta\partial z^\gamma},	\qquad 1\leq \alpha,\beta,\gamma\leq n.	\label{eq:frobenius_potential}
		\end{equation}
		The potential $F$ satisfies the WDVV associativity equations
		\begin{equation}	c_{\alpha\beta}^{\xi} c_{\xi\gamma}^{\delta} = c_{\gamma\beta}^{\xi} c_{\xi\alpha}^{\delta},	\qquad 1\leq \alpha,\beta,\gamma,\delta\leq n,	\label{eq:wdvv}
		\end{equation}
		together with the quasi-homogeneity condition
		\begin{equation}\label{quasi-homogeneity}	\mathcal{L}_E F = (3-d)F + A_{\alpha\beta}z^\alpha z^\beta + B_\alpha z^\alpha + C,
		\end{equation}
		where $A_{\alpha\beta}, B_\alpha, C$ are certain constants. The structure constants of the Frobenius multiplication are given by
		\begin{equation}
				\frac{\partial}{\partial z^\alpha} \cdot \frac{\partial}{\partial z^\beta}	= c_{\alpha\beta}^\gamma \frac{\partial}{\partial z^\gamma},	\label{eq:frobenius_structure_constants}
				\end{equation}
				where $c_{\alpha\beta}^\gamma = \eta^{\gamma\delta} c_{\alpha\beta\delta}$ and $c_{\alpha\beta\delta} = c\bigl(\frac{\partial}{\partial z^\alpha},\frac{\partial}{\partial z^\beta},\frac{\partial}{\partial z^\delta}\bigr)$. The Euler vector field $E$ can be expressed in flat coordinates as
				\begin{equation}
						E = \sum_{\alpha=1}^n \Bigl[ \Bigl(1-\frac{d}{2}-\mu_\alpha\Bigr)z^\alpha + r^\alpha \Bigr] \frac{\partial}{\partial z^\alpha},
						\end{equation}
						where the parameter $\mu := \operatorname{diag}(\mu_1,\dots,\mu_n)$ is defined by
						\begin{equation}
								\mu = 1-\frac{d}{2} - \nabla E,
						\end{equation}
						and satisfies the anti-commutation relation
						\begin{equation}
								\mu\eta + \eta\mu = 0.
								\end{equation}
								The constants $r^\alpha$ are non-zero only if $\mu_\alpha = 1-\frac{d}{2}$. 

    To study the Frobenius manifold associated with the $(3,1)$-type gAL hierarchy, we introduce the Landau--Ginzburg superpotential
    \begin{equation}\label{superpotential}
    	\lambda(p) = \frac{p^3 + u_1 p^2 + u_2 p + u_3}{p(p + w)}
    \end{equation}
    associated with the Lax operator \eqref{Lax-operator}. This superpotential endows the $(u_1,u_2,u_3,w)$-space with a Frobenius structure $M$, following the approach of \cite{Brini 2012}. More precisely, the Frobenius metric $\eta$ and the $(0,3)$-tensor $c$ that defines the Frobenius multiplication are given by the Landau--Ginzburg formulas \cite{Dubrovin1996}:
    \begin{align}
    	\eta(\p_\afa, \p_\beta) &=\,
    	\sum_{p_0\in C_\lmd}
    	\Res_{p=p_0}
    	\left(
    	\frac{\p_\afa(\lmd(p))\cdot \p_\beta(\lmd(p))}{\lmd'(p)}
    	\frac{\td p}{p^2}
    	\right),
    	\label{Frob-eta}\\
    	c(\p_\afa, \p_\beta, \p_\gamma) &=\,
    	\sum_{p_0\in C_\lmd}
    	\Res_{p=p_0}
    	\left(
    	\frac{\p_\afa(\lmd(p))\cdot \p_\beta(\lmd(p))\cdot \p_\gamma(\lmd(p))}{\lmd'(p)}
    	\frac{\td p}{p^2}
    	\right),  \label{Frob-c}
    \end{align}
    where $\p_\afa:=\pp{z^\afa}$, and $\p_\afa, \p_\beta, \p_\gamma$ denote vector fields on the Frobenius manifold $M$, while $C_\lmd=\Bigset{p}{\lmd'(p)=0}$ denotes the set of critical points of $\lmd(p)$.
    
    We now introduce a system of flat coordinates $(z^1,z^2,z^3,z^4)$ defined by
    \begin{align}\label{coordinate change}
    	\left\{
    	\begin{aligned}
    		z^1 &= - \underset{p=\infty}{\Res} \lambda(p) \frac{\mathrm{d}p}{p} = u_{1}-w, \\
    		z^2 &= \log u_{3}-2\log w, \\
    		z^3 &= - \underset{p=0}{\Res} \lambda(p) \frac{\mathrm{d}p}{p} =\frac{u_{2}}{w} -\frac{u_{3}}{w^{2}},\\
    		z^4 &= \log w.
    	\end{aligned}
    	\right.
    \end{align}
    In this coordinate system, the inverse metric $(\eta^{\afa\beta})=(\eta_{\afa\beta})^{-1}$ and the structure constants computed from \eqref{Frob-eta}--\eqref{Frob-c} take the simple forms
    \begin{equation}
    	(\eta_{\afa\beta})=(\eta^{\afa\beta}) =
    	\begin{pmatrix}
    		&	&   & 1 \\
    		&	& 1 &   \\
    		&	1 &   &\\
    		1 & & &
    	\end{pmatrix},
    \end{equation}
    \begin{equation}\label{c-ijk}
    	\left\{
    	\begin{split}
    		c_{111}&=\,\frac{1}{z^1-z^3},\quad
    		c_{112}=0,\quad
    		c_{113}=-\frac{1}{z^1-z^3},\quad
    		c_{114}=1,\\	c_{122}&=\rme^{z^2},\quad
    		c_{123}=\,0,\quad
    		c_{124}=0,\quad	c_{133}=\frac{1}{z^1-z^3},\\ c_{134}&=0,\quad c_{144}=-\rme^{z^4},\quad
    		c_{222}=\,\rme^{z^2}(\rme^{z^4}+z^1-z^3),\quad
    		c_{223}=\,-\rme^{z^2},\\
    		c_{224}&=\rme^{z^2+z^4},\quad c_{233}=1,\quad c_{234}=0,\quad
    		c_{244}=\rme^{z^2+z^4},\\	c_{333}&=-\frac{1}{z^1-z^3},\quad c_{334}=0,\quad c_{344}=\rme^{z^4},\quad c_{444}=\rme^{z^4}(\rme^{z^2}-z^1+z^3).
    	\end{split}\right.
    \end{equation}
    
    The potential $F$ of this Frobenius manifold is
    \begin{equation}\label{F-GAL}
    	F=\frac12(z^1)^2z^4
    	+\frac12 z^2 (z^3)^2
    	+\frac{1}{2}(z^1-z^3)^2\log (z^1-z^3)+(\rme^{z^2}-\rme^{z^4})(z^1-z^3)+\rme^{z^2+z^4},
    \end{equation}
    which satisfies $c_{\afa\beta\gamma}=\frac{\p^3 F}{\p z^\alpha \p z^{\beta}\p z^{\gamma}}$ and the quasi-homogeneity condition \eqref{quasi-homogeneity}. The Euler vector field $E$ and the charge $d$ are
    \begin{equation}\label{Euler-E}
    	E=z^1\p_{1}+\p_{2} + z^3\p_{3}+\p_{4},
    	\qquad
    	d=1.
    \end{equation}
    Moreover, the unity $e$ with respect to the Frobenius multiplication is
    \begin{equation} \label{unity e}
    	e=\p_{1}+\p_{3},
    \end{equation}
    which can be represented as a gradient field $e=\grad_{\eta}\theta_{0,0}$, where
    \begin{equation}\label{theta00}
    	\theta_{0,0} = z^2+z^4=\log \frac{u_{3}}{w}.
    \end{equation}
    
    It can be verified that this Frobenius manifold $M(\eta,c,e,E)$ satisfies all the axioms of a Frobenius manifold (see \cite{Frobenius manifold,Dubrovin1996} for details). Moreover, the intersection form
    \begin{equation}
    	g=g^{\afa\beta}\pp{z^\afa}\otimes\pp{z^\beta},\qquad
    	g^{\afa\beta}=E^\gamma c^{\afa\beta}_\gamma
    \end{equation}
    can be represented as
    \begin{equation}
    	(g^{\afa\beta})
    	=
    	\begin{pmatrix}
    		2\rme^{z^4}(\rme^{z^2}-z^1+z^3)
    		& \rme^{z^4}&2\rme^{z^2 +z^4}
    		& z^1-\rme^{z^4} \\[4pt]
    		\rme^{z^4}
    		& 2 & z^3-\rme^{z^2}&-1\\[4pt]
    		2\rme^{z^2 +z^4}
    		& z^3-\rme^{z^2}& 2\rme^{z^2}(\rme^{z^4}+z^1-z^3)&\rme^{z^2}\\[4pt]
    		z^1-\rme^{z^4}&-1&\rme^{z^2}&2
    	\end{pmatrix},
    \end{equation}
    and the contravariant metrics $\eta^{\afa\beta}, g^{\afa\beta}$ form a flat pencil on $M$. In terms of the original variables $(u^1,u^2,u^3,u^4):=(u_1,u_2,u_3,w)$, the coefficient matrices of this flat pencil coincide with \eqref{flat pencil}, which are also the leading coefficients of the Hamiltonian operators $\mcalP_1$ and $ \mcalP_2$ in \eqref{HamP1}--\eqref{HamP2}.
    	
    	\subsection{The Principal Hierarchy of $M$}
    	We will construct the Principal Hierarchy of the Frobenius manifold $M$ \eqref{F-GAL} in this subsection and show that the dispersionless limits of the first flows of the $(3,1)$-type gAL hierarchy belong to this Principal Hierarchy.
    	
    	In general, let $M(\eta,c,e,E)$ be an $n$-dimensional Frobenius manifold and fix a system of flat coordinates $(z^1,\dots,z^n)$. The Principal Hierarchy of $M$ is a family of evolutionary PDEs of the form
    	\begin{equation}\label{PH}
    		\frac{\partial z^\alpha}{\partial t^{i,k}} = \eta^{\alpha\beta} \partial_x \frac{\partial \theta_{i,k+1}}{\partial z^\beta}, \quad (i,k) \in \mathcal{I}
    	\end{equation}
    	in the time variables $\mathbf{t} = \{t^{i,k}\}_{(i,k) \in \mathcal{I}}$, where the index set is
    	\begin{align*}
    		\mathcal{I} = \bigl( \{1,2,\dots,n\} \times \mathbb{Z}_{\geq 0} \bigr) \cup \bigl( \{0\} \times \mathbb{Z} \bigr),
    	\end{align*}
    	and the functions $\theta_{i,k}$ on $M$ satisfy the initial conditions, recursion relations, and quasi-homogeneity conditions
    	\begin{equation}\label{4.16}
    		\theta_{\alpha,0} = \eta_{\alpha\beta} z^\beta, \quad \grad_{\eta} \theta_{0,0} = e,
    	\end{equation}
    	\begin{equation}\label{4.17}
    		\partial_\alpha \partial_\beta \theta_{i,k+1} = c_{\alpha\beta}^\gamma \partial_\gamma \theta_{i,k},\quad c_{\alpha\beta}^\gamma=\eta^{\gamma\lambda }c_{\lambda\alpha\beta},
    	\end{equation}
    	\begin{equation}\label{4.18}
    		\mathcal{L}_E \theta_{i,k} = \left( k + \mu_i + \frac{2-d}{2} \right) \theta_{i,k} + \sum_{s=1}^k \theta_{\alpha,k-s} \left( R_s \right)_i^\alpha (-1)^k r_{k+1}^\alpha \eta_{\alpha i},
    	\end{equation}
    	for $(i,k) \in \mathcal{I}$, where all lower Greek indices are assumed to range over $\{1,2,\dots,n\}$. 
    	
    	The flows of the Principal Hierarchy \eqref{PH} can be rewritten as a Hamiltonian system of hydrodynamic type
    	\begin{equation}
    		\frac{\partial z^\alpha}{\partial t^{i,k}} = (\mathcal{P}_1^{[0]})^{\alpha\beta} \frac{\delta H_{i,k}^{[0]}}{\delta z^\beta}, \quad (i,k) \in \mathcal{I},
    	\end{equation}
    	with compatible Hamiltonian operators $\{\mathcal{P}_1^{[0]}, \mathcal{P}_2^{[0]}\}$ of the form
    	\[
    	(\mathcal{P}_1^{[0]})^{\alpha\beta} = \eta^{\alpha\beta} \partial_x, \qquad
    	(\mathcal{P}_2^{[0]})^{\alpha\beta} = g^{\alpha\beta} \partial_x + \Gamma^{\alpha\beta}_{\gamma} z_x^\gamma,
    	\]
    	where the Hamiltonians are
    	\[
    	H_{i,k-1}^{[0]} = \int \theta_{i,k}(\mathbf{z}(x)) \, \td x, \quad (i,k) \in \mathcal{I}.
    	\]
    	Here $g^{\alpha\beta} = E^\gamma c_\gamma^{\alpha\beta}$ is the intersection form and $\Gamma_\gamma^{\alpha\beta} = \bigl(\tfrac12 - \mu_\beta\bigr) c_\gamma^{\alpha\beta}$ are the contravariant coefficients of the Levi-Civita connection of $(g^{\alpha\beta})$. More details can be found in \cite{Normal forms,GFM}.
    	
    	In the special case of the Frobenius manifold $M$ \eqref{F-GAL}, the first few terms of $\{\theta_{i,k}\}$ can be obtained by solving \eqref{4.16}--\eqref{4.18} directly:
    	\begin{align}
    		\theta_{1,0} &= z^4, \quad
    		\theta_{1,1} = \rme^{z^2}-\rme^{z^4}+z^3-z^1+(z^1-z^3)\log(z^3-z^1)+z^{1}z^{4}, \\
    		\theta_{2,0} &= z^3, \quad
    		\theta_{2,1} =\rme^{z^2+z^4}+\rme^{z^2}(z^1-z^3)+\frac{(z^3)^2}{2}, \label{2,1}\\
    		\theta_{3,0} &= z^2, \quad
    		\theta_{3,1} =  \rme^{z^4}-\rme^{z^2}+z^1-z^3+(z^3-z^1)\log(z^1-z^3)+z^{2}z^{3}, \\ 
    		\theta_{4,0} &= z^1, \quad
    		\theta_{4,1} =\rme^{z^2+z^4}+\rme^{z^4}(z^3-z^1)+\frac{(z^1)^2}{2},
    	\end{align}
    	and $\theta_{0,0}$ is given by \eqref{theta00}.
    	
    	By a straightforward calculation, the first few flows of the Principal Hierarchy~\eqref{PH} of $M$ are expressed in terms of the original variables $(u_1,u_2,u_3,w)$ via the coordinate change \eqref{coordinate change} as
    	\begin{align}
    		\frac{\partial}{\partial t^{2,0}} \begin{pmatrix} u_1 \\ u_2 \\ u_3 \\ w \end{pmatrix}
    		&=\frac{1}{w^2}
    		\begin{pmatrix}
    			0 & 0 & 2w & -3u_3 \\
    			u_3 w & 0 & u_1 w & -2u_1 u_3\\
    			0 & u_3 w & 0 & -u_2 u_3\\
    			0 & 0 & w & -2u_3
    		\end{pmatrix}
    		\begin{pmatrix} u_{1x} \\ u_{2x} \\ u_{3x} \\ w_x \end{pmatrix},
    		\label{250909-PH-1}\\[0.8ex]
    		\frac{\partial}{\partial t^{4,0}} \begin{pmatrix} u_1 \\ u_2 \\ u_3 \\ w \end{pmatrix}
    		&=
    		\begin{pmatrix}
    			0 & 1 & 0 & -u_{1} \\
    			u_{2} & 0 & 1 & -2u_{2}\\
    			2u_{3} & 0 & 0 & -3u_{3}\\
    			0 & 1 & 0 & -u_{1}
    		\end{pmatrix}
    		\begin{pmatrix} u_{1x} \\ u_{2x} \\ u_{3x} \\ w_x \end{pmatrix}.
    		\label{250909-PH-3}
    	\end{align}
    	
    	\begin{prop}\label{prop-250909}
    		The densities $\{\theta_{i,k}\}_{(i,k)\in\mathcal{I}}$ of $M$ \eqref{F-GAL} can be expressed as
    		\begin{align}
    			\theta_{2,k} &= \frac{1}{(k+1)!}\res_{p=0}\left(\lambda^{k+1}\frac{\mathrm{d}p}{p}\right), \label{theta2k}\\
    			\theta_{4,k} &= -\frac{1}{(k+1)!}\res_{p=\infty} \left(\lambda^{k+1}\frac{\mathrm{d}p}{p}\right), \label{theta4k}
    		\end{align}
    		for each $k\geq 0$.
    	\end{prop}
    	
    	\begin{proof}
    		Let $\Theta_{2,k}$ and $\Theta_{4,k}$ denote the right-hand sides of \eqref{theta2k} and \eqref{theta4k}, respectively. Using the coordinate change \eqref{coordinate change}, we obtain
    		\begin{align*}
    			\Theta_{2,1} &= \frac{1}{2}\res_{p=0}\left(\lambda^{2}\frac{\mathrm{d}p}{p}\right)
    			=-3w^2+4u_1w-u_1^2-2u_2
    			=\rme^{z^2+z^4}+\rme^{z^2}(z^1-z^3)+\frac{(z^3)^2}{2},
    		\end{align*}
    		which coincides with $\theta_{2,1}$ in \eqref{2,1}. Similarly, we can verify directly that
    		\begin{align*}
    			\Theta_{2,0}=\theta_{2,0},\quad
    			\Theta_{4,0}=\theta_{4,0},\quad
    			\Theta_{4,1}=\theta_{4,1}.
    		\end{align*}
    		
    		To prove the validity of \eqref{theta2k}--\eqref{theta4k}, we show that $\Theta_{i,k}$ satisfies \eqref{4.17}--\eqref{4.18}, i.e.,
    		\begin{equation}\label{4.40}
    			\partial_\alpha \partial_\beta \Theta_{i,k+1} = c^\gamma_{\alpha\beta} \partial_\gamma \Theta_{i,k},
    		\end{equation}
    		where the Frobenius coefficients $c_{\alpha\beta\gamma}$ of $M$ are given in \eqref{c-ijk} and $c^\gamma_{\alpha\beta} = \eta^{\gamma\delta} c_{\alpha\beta\delta}$.
    		
    		Note that
    		\[
    		\partial_\alpha \partial_\beta \Theta_{2,k+1}
    		=\frac{1}{(k+1)!} \res_{p=0}\Bigl(
    		\bigl(\partial_\alpha\partial_\beta\lambda \cdot \lambda^{k+1}
    		+ \partial_\alpha\lambda\,\partial_\beta\lambda \cdot (k+1)\lambda^{k}\bigr)\frac{\mathrm{d}p}{p}\Bigr).
    		\]
    		A direct calculation then gives
    		\begin{align*}
    			\partial_1 \partial_1 \Theta_{2,k+1} - c^\alpha_{11} \partial_\alpha \Theta_{2,k} &= 0, \\
    			\partial_3 \partial_3 \Theta_{2,k+1} - c^\alpha_{33} \partial_\alpha \Theta_{2,k} &= 0.
    		\end{align*}
    		The remaining cases of \eqref{4.40} can be verified similarly, so we omit the details. In addition, for all $k \geq 0$ the following relations hold:
    		\begin{align}\label{4.41}
    			\mathcal{L}_E \Theta_{2,k} = (k+1) \Theta_{2,k},\qquad
    			\mathcal{L}_E \Theta_{4,k} = (k+1) \Theta_{4,k},
    		\end{align}
    		with $E$ defined in \eqref{Euler-E}. This proves the proposition.
    	\end{proof}
     \subsection{Examples of Legendre transformations of $M$}\label{subsection:Legendre transf}
  We first recall the necessary preliminaries on Legendre transformations of Frobenius manifolds \cite{LiuQuZhang,GLt}. Let $M(\eta, c, e, E)$ be an $n$-dimensional Frobenius manifold with flat metric $\eta$, Frobenius multiplication $c$, unity $e$, and Euler vector field $E$ of charge $d$. A vector field $B \in \operatorname{Vect}(M)$ is called a Legendre field if it satisfies
 \[
 X_1 \cdot \nabla_{X_2} B = X_2 \cdot \nabla_{X_1} B, \qquad \forall X_1, X_2 \in \operatorname{Vect}(M),
 \]
 where $\nabla$ denotes the Levi-Civita connection of $\eta$ and $X_1 \cdot X_2 := c(X_1, X_2)$ denotes the Frobenius multiplication. A Legendre field $B$ is invertible if there exists a vector field $B^{-1} \in \operatorname{Vect}(M)$ such that $B \cdot B^{-1} = e$. For any invertible Legendre field $B$, we define a new metric $\hat{\eta}$ on $M$ by
 \begin{equation}\label{hat-metric-def}
 	\hat{\eta}(X_1, X_2) = \eta(B \cdot X_1, B \cdot X_2), \qquad \forall X_1, X_2 \in \operatorname{Vect}(M).
 \end{equation}
 A Legendre field $B$ is said to be quasi-homogeneous if there exists a constant $\mu_B \in \mathbb{C}$ such that
 \begin{equation*}
 	[E, B] = \Bigl(\mu_B - 1 + \frac{d}{2}\Bigr) B.
 \end{equation*}
 If $B$ is a quasi-homogeneous, invertible Legendre field, then $\hat{\eta}$ is also flat on $M$  \cite{LiuQuZhang} and satisfies the quasi-homogeneity condition
 \[
 \mathcal{L}_E\hat{\eta} = (2 - \hat{d})\hat{\eta}, \qquad \hat{d} = -2\mu_B.
 \]
 Consequently, $\hat{M} := M(\hat{\eta}, c, e, E)$ is a (generalized) Frobenius manifold of charge $\hat{d}$, called the generalized Legendre transformation of $M$ along $B$. The new manifold $\hat{M}$ inherits the same Frobenius multiplication $c$, unity $e$, and Euler vector field $E$ from $M$; only the flat metric is changed.
 
 Moreover, the inverse $\hat{B} := B^{-1}$ is a quasi-homogeneous Legendre field on $\hat{M}$ satisfying
 \[
 [E, \hat{B}] = -\Bigl(\mu_B +1+ \frac{d}{2}\Bigr)\hat{B},
 \]
 and transforms $\hat{M}$ back to $M$. In this sense, the pairs $(M, B)$ and $(\hat{M}, \hat{B})$ are said to be Legendre dual to each other.
 
 Fix a system of flat coordinates $(z^1,\dots,z^n)$ for the metric $\eta$ on $M$, and let $B = B^\alpha \frac{\partial}{\partial z^\alpha}$ be a quasi-homogeneous, invertible Legendre field. Then there exists a system of flat coordinates $(\hat{z}^1,\dots,\hat{z}^n)$ for the transformed metric $\hat{\eta}$ such that
 \begin{equation}\label{relations4.43}
 	\frac{\partial \hat{z}^\alpha}{\partial z^\beta} = B^\gamma c_{\beta \gamma}^\alpha, \qquad
 	\eta\!\left( \frac{\partial}{\partial z^\alpha}, \frac{\partial}{\partial z^\beta} \right)
 	= \hat{\eta}\!\left( \frac{\partial}{\partial \hat{z}^\alpha}, \frac{\partial}{\partial \hat{z}^\beta} \right),
 \end{equation}
 where $c_{\beta \gamma}^\alpha$ are the structure constants of the Frobenius multiplication in the $z^{\afa}$-coordinates. Furthermore, the prepotentials $F$ of $M$ and $\hat{F}$ of $\hat{M}$ satisfy
 \begin{equation}\label{prepotential-relation}
 	\frac{\partial^2 F}{\partial z^\alpha \partial z^\beta} = \frac{\partial^2 \hat{F}}{\partial \hat{z}^\alpha \partial \hat{z}^\beta}.
 \end{equation}
 When both the unity $e$ and the Legendre field $B$ are flat, i.e., $\nabla e = 0$ and $\nabla B = 0$, the generalized Legendre transformation reduces to the classical Legendre transformation originally introduced by Dubrovin \cite{Dubrovin1996}.
 
 We now apply this theory to the Frobenius manifold $M$ given by \eqref{F-GAL}.
 
 \begin{ex}
 	Consider the quasi-homogeneous, invertible Legendre field
 	\[
 	B_1=\frac{\partial}{\partial z^2}
 	\]
 	on $M$~\eqref{F-GAL}. From the relations \eqref{relations4.43} we obtain the new flat coordinates $\hat z^\alpha$ as
 	\begin{equation}\label{Legendre change1}
 		\left\{
 		\begin{aligned}
 			\hat{z}^1 &=  \rme^{z^2+z^4}, \\[4pt]
 			\hat{z}^2 &= z^3 -  \rme^{z^2}, \\[4pt]
 			\hat{z}^3 &=\rme^{z^2+z^4}+ \rme^{z^2}(z^1-z^3),\\[4pt]
 			\hat{z}^4 &=\rme^{z^2},
 		\end{aligned}
 		\right.
 		\quad
 		\big( \hat{\eta}(\partial_{\hat{z}^\alpha}, \partial_{\hat{z}^\beta}) \big) =
 		\begin{pmatrix}
 			0&	0&0 & 1 \\[4pt]
 			0&	0&1 & 0\\[4pt]
 			0&	1&0&0\\[4pt]
 			1&	0&0&0
 		\end{pmatrix}.
 	\end{equation}
 	Using \eqref{prepotential-relation}, we obtain the potential $\hat{F}$ of $\hat{M} = M(\hat{\eta}, c, e, E)$ 
 	\begin{align}\label{F1}
 		\hat{F} =& \hat{z}^{1}\hat{z}^{2}\hat{z}^4+\frac{1}{2} \hat{z}^1(\hat{z}^4)^2  + \frac{1}{2} (\hat{z}^2)^2\hat{z}^3  + \frac{1}{2} (\hat{z}^3-\hat{z}^1)^2 \log(\hat{z}^3-\hat{z}^1) +\hat{z}^1 \hat{z}^3 \log \hat{z}^4\\
 		&+\frac{1}{2}(\hat{z}^1)^2 \log(-\hat{z}^1)-(\hat{z}^1)^2 \log \hat{z}^4 \nonumber.
 	\end{align}
 	The Euler field $E$ and the unity $e$ have the form
 	\begin{equation*}
 		E = 2\hat{z}^1 \frac{\partial}{\partial \hat{z}^1} +  \hat{z}^2 \frac{\partial}{\partial \hat{z}^2} + 2\hat{z}^3 \frac{\partial}{\partial \hat{z}^3}+ \hat{z}^4 \frac{\partial}{\partial \hat{z}^4}, \quad e = \frac{\partial}{\partial \hat{z}^2}
 	\end{equation*}
 	with respect to the new coordinates $\hat z^\alpha$.
 \end{ex}

 \begin{ex}
 	Consider the quasi-homogeneous, invertible Legendre field
 	\[
 	B_2=\frac{\partial}{\partial z^4}
 	\]
 	on $M$~\eqref{F-GAL}. From the relations \eqref{relations4.43} we obtain the new flat coordinates $\tilde z^\alpha$ as
 	\begin{equation}\label{Legendre change2}
 		\left\{
 		\begin{aligned}
 			\tilde{z}^1 &=\rme^{z^2+z^4}- \rme^{z^4}(z^1-z^3), \\[4pt]
 			\tilde{z}^2 &=\rme^{z^4}, \\[4pt]
 			\tilde{z}^3 &=\rme^{z^2+z^4},\\[4pt]
 			\tilde{z}^4 &= z^1-  \rme^{z^4},
 		\end{aligned}
 		\right.
 		\quad
 		\big( \tilde{\eta}(\partial_{\tilde{z}^\alpha}, \partial_{\tilde{z}^\beta}) \big) =
 		\begin{pmatrix}
 			0&	0&0 & 1 \\[4pt]
 			0&	0&1 & 0\\[4pt]
 			0&	1&0&0\\[4pt]
 			1&	0&0&0
 		\end{pmatrix}.
 	\end{equation}
 	Using \eqref{prepotential-relation}, we obtain the potential $\tilde{F}$ of $\tilde{M} = M(\tilde{\eta}, c, e, E)$ 
 	\begin{align}\label{F2}
 		\tilde{F} =&\frac{1}{2} (\tilde{z}^2)^2\tilde{z}^3  -\tilde{z}^{1}\tilde{z}^{2}\tilde{z}^4+2\tilde{z}^2\tilde{z}^3\tilde{z}^4+\frac{1}{2} \tilde{z}^3(\tilde{z}^4)^2   + \frac{1}{2} (\tilde{z}^3)^2 \log \tilde{z}^3-\frac{1}{2} (\tilde{z}^3)^2 \log \tilde{z}^2 \nonumber\\
 		& +\tilde{z}^1 \tilde{z}^2 \tilde{z}^4\log \tilde{z}^4-\tilde{z}^2 \tilde{z}^3 \tilde{z}^4\log \tilde{z}^4 + \frac{1}{2}(\tilde{z}^3-\tilde{z}^1)^2 \log(\tilde{z}^3-\tilde{z}^1)\nonumber\\
 		&+\frac{1}{2}(\tilde{z}^1)^2 \log(\tilde{z}^3-\tilde{z}^1)-\frac{1}{2}(\tilde{z}^1)^2 \log(\tilde{z}^1-\tilde{z}^3).
 	\end{align}
 	The Euler field $E$ and the unity $e$ have the form
 	\begin{equation*}
 		E = 2\tilde{z}^1 \frac{\partial}{\partial \tilde{z}^1} +  \tilde{z}^2 \frac{\partial}{\partial \tilde{z}^2} + 2\tilde{z}^3 \frac{\partial}{\partial \tilde{z}^3}+ \tilde{z}^4 \frac{\partial}{\partial \tilde{z}^4}, \quad e = \frac{\partial}{\partial \tilde{z}^4}.
 	\end{equation*}
 \end{ex}
 
 \begin{ex}
 	Consider the quasi-homogeneous, invertible Legendre field
 	\[
 	B_3=\frac{\partial}{\partial z^2}-\frac{\partial}{\partial z^4}
 	\]
 	on $M$~\eqref{F-GAL}. From the relations \eqref{relations4.43} we obtain the new flat coordinates $\bar z^\alpha$ as
 	\begin{equation}\label{Legendre change3}
 		\left\{
 		\begin{aligned}
 			\bar{z}^1 &=  \rme^{z^4}(z^1-z^3), \\[4pt]
 			\bar{z}^2 &= z^3 -  \rme^{z^2}-  \rme^{z^4}, \\[4pt]
 			\bar{z}^3 &=\rme^{z^2}(z^1-z^3),\\[4pt]
 			\bar{z}^4 &=\rme^{z^2}+\rme^{z^4}-z^1,
 		\end{aligned}
 		\right.
 		\quad
 		\big( \bar{\eta}(\partial_{\bar{z}^\alpha}, \partial_{\bar{z}^\beta}) \big) =
 		\begin{pmatrix}
 			0&	0&0 & 1 \\[4pt]
 			0&	0&1 & 0\\[4pt]
 			0&	1&0&0\\[4pt]
 			1&	0&0&0
 		\end{pmatrix}.
 	\end{equation}
 	Using \eqref{prepotential-relation}, we obtain the potential $\bar{F}$ of $\bar{M} = M(\bar{\eta}, c, e, E)$ 
 	\begin{equation}\label{F3}
 		\bar{F} =\frac{1}{2} (\bar{z}^2)^2\bar{z}^3  -\frac{1}{2} \bar{z}^1(\bar{z}^4)^2  +\frac{1}{2}(\bar{z}^1)^2 \log(\bar{z}^1)+\frac{1}{2}(\bar{z}^3)^2 \log(\bar{z}^3)-\bar{z}^1\bar{z}^3 \log(\bar{z}^2+\bar{z}^4).
 	\end{equation}
 	The Euler field $E$ and the unity $e$ have the form
 	\begin{equation}\label{Legndre3}
 		E = 2\bar{z}^1 \frac{\partial}{\partial \bar{z}^1} +  \bar{z}^2 \frac{\partial}{\partial \bar{z}^2} + 2\bar{z}^3 \frac{\partial}{\partial \bar{z}^3}+ \bar{z}^4 \frac{\partial}{\partial \bar{z}^4}, \quad e = \frac{\partial}{\partial \bar{z}^2}-\frac{\partial}{\partial \bar{z}^4}
 	\end{equation}
 	with respect to the new coordinates $\bar z^\alpha$.
 	
 	We introduce new coordinates \(\boldsymbol{s} = (s^1, s^2, s^3, s^4)\) to transform the unity \(e = \frac{\partial}{\partial \bar{z}^2} - \frac{\partial}{\partial \bar{z}^4}\) obtained after the above Legendre transformation into the form \(e = \frac{\partial}{\partial s^1}\). To this end, we set
 	\begin{equation}\label{Legendre-s-coords}
 		\left\{
 		\begin{aligned}
 			s^1 &=\bar{z}^2, \\[4pt]
 			s^2 &=\bar{z}^1, \\[4pt]
 			s^3 &=\bar{z}^3,\\[4pt]
 			s^4 &=\bar{z}^2+\bar{z}^4,
 		\end{aligned}
 		\right.
 		\quad
 		\left(\frac{\partial s^i}{\partial \bar{z}^j} \right) =
 		\begin{pmatrix}
 			0&	1&0 & 0 \\[4pt]
 			1&	0&0 & 0\\[4pt]
 			0&	0&1&0\\[4pt]
 			0&	1&0&1
 		\end{pmatrix}.
 	\end{equation}
 	The potential $\bar{F}$ given in \eqref{F3} and the Euler vector field $E$ given in \eqref{Legndre3} then take the following forms in the $\boldsymbol{s}$-coordinates:
 	\begin{align}
 		\bar{F} =&\frac{1}{2} (s^1)^2s^3 -\frac{1}{2} s^2(s^4)^2  -\frac{1}{2}(s^1)^2 s^2+s^1s^2s^4 \nonumber\\
 		&+\frac{1}{2}(s^2)^2 \log(s^2)+\frac{1}{2}(s^3)^2 \log(s^3)-s^2s^3 \log(s^4),
 	\end{align}
 	and
 	\begin{equation}
 		E = s^1 \frac{\partial}{\partial s^1} +  2s^2 \frac{\partial}{\partial s^2} + 2s^3 \frac{\partial}{\partial s^3}+ s^4 \frac{\partial}{\partial s^4}.
 	\end{equation}
 \end{ex}
 \section{Conclusion}  
In this work, we show that the $(3,1)$-type gAL hierarchy possesses a local tri-Hamiltonian structure at the full-dispersive level, realized explicitly via the supervariable technique. Moreover, the central invariants of the bi-Hamiltonian structure $(\mathcal P_1,\mathcal P_2)$ are all equal to $\frac{1}{24}$, which implies the existence of a tau function for this integrable system and paves the way for further studies on super tau-covers, Virasoro symmetry, and related topics. We further construct an associated Frobenius manifold $M$ and show that the dispersionless limits of the first flows of the $(3,1)$-type gAL hierarchy belong to the Principal Hierarchy of $M$, which helps to elucidate the intrinsic connection between integrable systems and 2D topological field theory.
 
The gAL hierarchy has been generalized to the rational reduction of the 2D-Toda hierarchy (RR2T) by Brini and his collaborators \cite{RR2T} in 2017. Qu and the second author of this paper recently studied the local bi-Hamiltonian structures of the $(1,2)$- and $(2,1)$-type RR2T, together with a generalized Frobenius manifold with non-flat unity for the $(2,1)$-type case \cite{2-1 RR2T}, and further extended these results to the $(1,n)$- and $(n,1)$-type RR2T \cite{n-1 RR2T}. One readily sees that the $(1,n)$-type gAL hierarchy corresponds to the $(1,n)$-type RR2T, i.e., the case $(a,b)=(1,n)$ in \eqref{ab type gAL}. However, why does the present work bypass the $(2,1)$-type case? Our preliminary investigation indicates that the $(2,1)$-type gAL hierarchy may not admit a local Hamiltonian structure. The $(3,1)$-type gAL hierarchy, by contrast, admits a tri-Hamiltonian structure as well as a Frobenius manifold with flat unity. These results suggest that the $(n,1)$-type gAL hierarchy differs from the $(n,1)$-type RR2T in certain respects. Building on the above, future work will extend our results to more general asymmetric gAL hierarchies and further investigate the symmetric case.
  \vskip 0.3cm
  \noindent \textbf{Acknowledgements.} The authors are grateful to the referees for their useful suggestions, which have improved this paper. The work was supported by the “Jingying” Project of Shandong University of Science and Technology.


\begin{thebibliography}{99}
	\bibitem{AL1975}M.J. Ablowitz, J.F. Ladik, Nonlinear differential-difference equations, J. Math. Phys. 16 (1975) 598--603.
	\bibitem{AL1976}M.J. Ablowitz, J.F. Ladik, Nonlinear differential-difference equations and Fourier analysis, J. Math. Phys. 17 (1976) 1011--1018.
	
	\bibitem{Brini CMP}A. Brini, The local Gromov--Witten theory of $\mathbb{CP}^1$ and integrable hierarchies, Commun. Math. Phys. 313 (2012) 571--605.
	\bibitem{Brini 2012}A. Brini, G. Carlet, P. Rossi, Integrable hierarchies and the mirror model of local $\mathbb{CP}^1$, Phys. D 241 (2012) 2156--2167.
	\bibitem{RR2T}A. Brini, G. Carlet, S. Romano, P. Rossi, Rational reductions of the 2D-Toda hierarchy and mirror symmetry, J. Eur. Math. Soc. 19 (2017) 835--880.
	
	\bibitem{unobstructed}G. Carlet, H. Posthuma, S. Shadrin, Deformations of semisimple Poisson pencils of hydrodynamic type are unobstructed, J. Differential Geom. 108 (2018) 63--89.
	\bibitem{Fan}M. Chen, J. He, E. Fan, Long-time asymptotics for the defocusing Ablowitz--Ladik system with initial data in lower regularity, Adv. Math. 450 (2024) 109769.
	
	\bibitem{TFT-2}R. Dijkgraaf, H. Verlinde, E. Verlinde, Topological strings in $d \leq 1$, Nucl. Phys. B 352 (1991) 59--86.
	\bibitem{Frobenius manifold}B. Dubrovin, Integrable systems and classification of 2-dimensional topological field theories, Progr. Math. 115 (1993) 313--359.
	\bibitem{Dubrovin1996}B. Dubrovin, Geometry of 2D topological field theories, In: Integrable systems and quantum groups, Springer Berlin (1996) 120--348.
	\bibitem{Normal forms}B. Dubrovin, Y. Zhang, Normal forms of hierarchies of integrable PDEs, Frobenius manifolds and Gromov--Witten invariants, eprint arXiv:math.DG/0108160.
	\bibitem{c1}B. Dubrovin, S.-Q. Liu, Y. Zhang, On Hamiltonian perturbations of hyperbolic systems of conservation laws I: quasi-triviality of bi-Hamiltonian perturbations, Comm. Pure Appl. Math. 59 (2006) 559--615.
	
	\bibitem{c3}G. Falqui, P. Lorenzoni, Exact Poisson pencils, $\tau$-structures and topological hierarchies, Physica D 241 (2012) 2178--2187.
	\bibitem{Ferapontov}E.V. Ferapontov, Compatible Poisson brackets of hydrodynamic type, J. Phys. A 34 (2001) 2377--2388.
	
	\bibitem{relativistic Toda2}J. Gibbons, B.A. Kupershmidt, Relativistic analogs of basic integrable systems, In: Integrable and Superintegrable Systems, World Sci. (1990) 207--231.
	\bibitem{Geng}X. Geng, H. Dai, J. Zhu, Decomposition of the discrete Ablowitz--Ladik hierarchy, Stud. Appl. Math. 118 (2007) 281--312.
	
	\bibitem{Kersten}P. Kersten, I. Krasil'shchik, A. Verbovetsky, Hamiltonian operators and $l^{*}$-coverings, J. Geom. Phys. 50 (2004) 273--302.
	
	\bibitem{AL-triham}S. Li, S.-Q. Liu, H. Qu, Y. Zhang, Tri-Hamiltonian structure of the Ablowitz--Ladik hierarchy, Phys. D 433 (2022) 133180.
	\bibitem{c2}S.-Q. Liu, Y. Zhang, Deformations of semisimple bihamiltonian structures of hydrodynamic type, J. Geom. Phys. 54 (2005) 427--453.
	\bibitem{Jacobi structure}S.-Q. Liu, Y. Zhang, Jacobi structures of evolutionary partial differential equations, Adv. Math. 227 (2011) 73--130.
	\bibitem{Bihamiltonian cohomologies}S.-Q. Liu, Y. Zhang, Bihamiltonian cohomologies and integrable hierarchies I: a special case, Commun. Math. Phys. 324 (2013) 897--935.
	\bibitem{Liu lecture notes}S.-Q. Liu, Lecture notes on bihamiltonian structures and their central invariants, B-Model Gromov--Witten Theory, Cham: Springer International Publishing (2019) 573--625.
	\bibitem{Liu2024}S.-Q. Liu, H. Qu, Y. Wang, Y. Zhang, Solutions of the loop equations of a class of generalized Frobenius manifolds, Commun. Math. Phys. 405 (2024) 225.
	\bibitem{extend AL}S.-Q. Liu, Y. Wang, Y. Zhang, The extended Ablowitz--Ladik hierarchy and a generalized Frobenius manifold, eprint arXiv:2404.08895.
	\bibitem{GFM}S.-Q. Liu, H. Qu, Y. Zhang, Generalized Frobenius manifolds with non-flat unity and integrable hierarchies, Commun. Math. Phys. 406 (2025) 77.
	\bibitem{LiuQuZhang}S.-Q. Liu, H. Qu, Y. Zhang, Legendre transformations of a class of generalized Frobenius manifolds and the associated integrable hierarchies, Commun. Math. Phys. 406 (2025) 121.
	
	
	\bibitem{relativistic Toda1}W. Oevel, B. Fuchssteiner, H. Zhang, O. Ragnisco, Mastersymmetries, angle variables, and recursion operator of the relativistic Toda lattice, J. Math. Phys. 30 (1989) 2664--2670.
	
	
	\bibitem{2-1 RR2T}H. Qu, Q. Zhao, Asymmetric rational reductions of 2D-Toda hierarchy and a generalized Frobenius manifold, eprint arXiv:2510.04151.
	
	\bibitem{n-1 RR2T} H. Qu, Q. Zhao, Bihamiltonian structure of the $(n,1)$-type rational reductions of the 2D-Toda hierarchy, eprint arXiv:2606.23167.
	\bibitem{GLt}I.A. Strachan, R. Stedman, Generalized Legendre transformations and symmetries of the WDVV equations, J. Phys. A 50 (2017) 095202.
	\bibitem{Suris}Y.B. Suris, The problem of integrable discretization: Hamiltonian approach, in: Progress in Mathematics, Vol. 219, Birkh\"auser Verlag, Basel, 2003.
	
	\bibitem{Takasaki GAL}K. Takasaki, Generalized Ablowitz--Ladik hierarchy in topological string theory, J. Phys. A 47 (2014) 165201.
	
	\bibitem{AL&2DT}V.E. Vekslerchik, The 2D Toda lattice and the Ablowitz--Ladik hierarchy, Inverse Problems 11 (1995) 463--479.
	\bibitem{Universality}V.E. Vekslerchik, `Universality' of the Ablowitz--Ladik hierarchy, eprint solv-int/9807005 (1998).
	
	\bibitem{TFT-1}E. Witten, On the structure of the topological phase of two-dimensional gravity, Nucl. Phys. B 340 (1990) 281--332.
	
	\bibitem{AL symmetries1}D. Zhang, S. Chen, Symmetries for the Ablowitz--Ladik hierarchy: Part I. Four-potential case, Stud. Appl. Math. 125 (2010) 393--418.
	\bibitem{Central invariants}Y. Zhang, Central invariants of semisimple bihamiltonian structures, in: Proceedings of the 4th International Congress of Chinese Mathematicians III, Higher Education Press (2008) 380--394.

	
\end{thebibliography}
\end{document}